\documentclass{jpp}
\usepackage{graphicx}

\usepackage[utf8]{inputenc}
\usepackage[T1]{fontenc}
\usepackage{amsmath}
\usepackage{comment}
\usepackage{layouts}
\usepackage{fancyhdr}

\newtheorem{proposition}{Proposition}

\shorttitle{Zonal Jets and Shearless Tori}
\shortauthor{N. M. Cao, H. Zhu, G. C. Grime, and T. Stoltzfus-Dueck}

\title{Detecting Shearless Phase-Space Transport Barriers in Global Gyrokinetic Turbulence Simulations with Test Particle Map Models}

\author{Norman M. Cao\aff{1}
  \corresp{\email{norman.cao@austin.utexas.edu}},
  Hongxuan Zhu\aff{2,3},
  Gabriel C. Grime\aff{1,4}
  \and Timothy Stoltzfus-Dueck\aff{5}}

\affiliation{\aff{1}Institute for Fusion Studies, The University of Texas at Austin, Austin, TX 78712, USA
\aff{2}{School of Physics, Zhejiang University, Hangzhou 310027, China}
\aff{3}Department of Astrophysical Sciences, Princeton University, Princeton, NJ 08544, USA
\aff{4}Institute of Physics, University of S\~{a}o Paulo, S\~{a}o Paulo, SP, 05508-220, Brazil
\aff{5}Princeton Plasma Physics Laboratory, Princeton, NJ 08540, USA}

\begin{document}

\maketitle

\begin{abstract}
In magnetically confined fusion plasmas, the role played by zonal \(E \times B\) flow shear layers in the suppression of turbulent transport is relatively well-understood.
However, less is understood about the role played by the weak shear regions that arise in the non-monotonic radial electric field profiles often associated with these shear layers.
In electrostatic simulations from the global total-\(f\) gyrokinetic particle-in-cell code XGC, we demonstrate how shearless regions with non-zero flow curvature form zonal ``jets'' that, in conjunction with neighboring regions of shear, can act as robust barriers to particle transport and turbulence spreading.
By isolating quasi-coherent fluctuations radially localized to the zonal jets, we construct a map model for the Lagrangian dynamics of gyrokinetic test particles in the presence of drift waves.
We identify the presence of shearless invariant tori in this model and verify that these tori act as partial phase-space transport barriers in the simulations.
We also demonstrate how avalanches impinging on these shearless tori cause eddy detachment events that form ``cold/warm core ring'' structures analogous to those found in oceanic jets, facilitating transport across the barriers without destroying them completely.
We discuss how shearless tori may generically arise from tertiary instabilities or other types of discrete eigenmodes, suggesting their potential relevance to broader classes of turbulent fluctuations.
\end{abstract}

\section{Introduction}

Sheared \(E \times B\) flows arising from the toroidally symmetric component of the radial electric field \(E_r\), also known as zonal flows, are central to turbulence regulation in tokamak plasmas \citep{diamond_zonal_2005,burrell_role_2020,staebler_quasilinear_2024}.
A key thrust in magnetic confinement fusion is thus to understand the mechanisms which determine the coupled intensities of zonal flows and turbulent fluctuations, which in turn predict the overall level of cross-field turbulent transport.
In the core, zonal flows are known to play a key role the nonlinear upshift of the critical gradient for the onset of turbulent transport, known as the ``Dimits shift'' \citep{dimits_simulation_2000}.
In the H-mode edge and other edge transport barrier regimes, the sheared flows are strong enough to completely quench or otherwise severely limit the amount of transport due to turbulence, leading to enhanced plasma confinement.
In the L-mode edge, the growth of sheared zonal flows likely plays a key role in triggering the L-H transition \citep{kim_zonal_2003,schmitz_role_2012}, while the collapse of sheared flows in the plasma edge has been proposed as a mechanism for setting the L-mode density limit \citep{singh_bounds_2021,diamond_how_2023}.

A key observation, however, is that shear layers rarely occur in isolation.
Near the last closed flux surface, \(E_r\) profiles in both L-mode and H-mode are typically observed with a well- or hill-like structure \citep{viezzer_high-accuracy_2013,grenfell_h-mode_2018}.
These correspond to zonal \(E \times B\) \textit{jets}, consisting of extrema of the zonal flow flanked by oppositely-signed shear layers, leading to a region of \textit{zero shear} in-between.
Morover, since \(E_r \propto \nabla p_i\) for ion-diamagnetic-dominated flows, the steepest gradient regions are typically found at local extrema of \(E_r\), hence near the shearless regions.
A naive application of local shear suppression criteria might suggest that these steepest gradient regions should correspond to regions of maximum transport.
While this is indeed sometimes observed, for example in ``staircase pedestals'' \citep{ashourvan_formation_2019}, this situation does not seem universal.
Similar behavior is relevant to the core, where shear layers in flux-driven global gyrokinetic simulations are often observed to organize into long-lived layered structures known as ``staircases'' which regulate the radial correlation length of turbulence \citep{dif-pradalier_validity_2010}.
Observe in Fig. 1 of \citet{dif-pradalier_e_2017} how the steepest temperature gradient and lowest turbulent heat flux region nearly aligns with a zero crossing of the \(E \times B\) shear.

In this work, we examine turbulence in non-degenerate shearless regions where the zonally-averaged shear passes through zero with a non-zero slope, corresponding to a local extremum in the zonal flow profile.
This non-zero slope is equivalent to a non-zero flow curvature in the shearless region.
`Shearless' is somewhat of a misnomer as the shear is only zero at an isolated radial location; whenever we use the term `shearless' in this work, we will implicitly assume non-degeneracy, and hence non-zero flow curvature.

There is a well-developed understanding of how local \(E \times B\) shear can lead to the saturation or suppression of turbulence \citep{biglari_influence_1990,hahm_flow_1995,waltz_theory_1998}.
Work in eikonal theory and wave kinetics has developed a detailed picture of how the radial variation of the zonal shear can affect drift wave turbulence in the scale-separated geometric optics level \citep{smolyakov_coherent_2000,ruiz_zonal-flow_2016}.
However, recent work has pointed out the role that effects beyond the geometric optics level play in drift waves-zonal flow interaction, in particular the theory of the tertiary instability \citep{zhu_theory_2020,zhu_theory_jpp_2020}.
Quantization of the tertiary instability modes can lead to a radial localization of the drift-wave envelopes which can stabilize instabilities at gradients above the critical gradient of the primary linear instability, which has been indicated as a potential mechanism for the nonlinear critical gradient associated with the Dimits shift \citep{kobayashi_quench_2012}.
Shearless regions play a key role in tertiary instablity theory, as they are associated with the X- and O-points in the drifton phase space \citep{sasaki_enhancement_2017,zhu_structure_2018} where tertiary instabilities can localize around.

The central theme of this work is to show that the localization of fluctuations to shearless regions can counterintuitively assist in suppression of turbulent transport through the formation of \textit{shearless transport barriers}.
In the dynamical systems literature, shearless transport barriers refer to a particular class of invariant tori present in Hamiltonian flows and maps \citep{del-castillo-negrete_chaotic_2000,morrison_magnetic_2000,caldas_shearless_2012}, and in this context are also known as shearless or nontwist tori.
Here, we use the terminology `shearless tori' to refer to invariant tori in collisionless test particle dynamics, `shearless transport barriers' to refer to their macroscopic manifestation in turbulent transport, and `shearless phase-space transport barriers' to refer to their manifestation in particle phase space.
In fixed-parameter planar maps or periodic planar flows invariant tori act as complete barriers to transport, but such conclusions do not immediately apply when considering transport due to non-autonomous perturbations, higher-dimensional phase space, and collisional effects experienced by particles in gyrokinetic turbulence.

The theory of shearless transport barriers originally arose to explain observations of robust barriers to the transport of dye across zonal jets in rotating tank experiments \citep{behringer_chaos_1991,del-castillo-negrete_chaotic_1993}.
Shearless tori have been extensively studied in a number of different systems, where they have been observed phenomenologically to be robust barriers to transport.
These systems include 2-dimensional models \citep{del-castillo-negrete_area_1996,del-castillo-negrete_renormalization_1997,balescu_hamiltonian_1998,del-castillo-negrete_chaotic_2000,marcus2008,da_fonseca_area-preserving_2014}, 3-dimensional models including magnetic shear and \(E \times B\) shear \citep{horton_drift_1998,marcus_influence_2019,osorio-quiroga_shaping_2023,grime_shearless_2023}, and in gyrokinetic models with prescribed perturbations \citep{anastassiou_role_2024}.
There also exist generalizations of the notion of shearless transport barriers to finite-time systems \citep{rypina_lagrangian_2007,beron-vera_invariant-tori-like_2010,farazmand_shearless_2014,falessi_lagrangian_2015}.

In this work, we show that the theory of shearless tori can be used to identify the presence of shearless transport barriers in high-fidelity global gyrokinetic simulations.
We demonstrate, for the first time in a global gyrokinetic simulation of electrostatic turbulence with realistic geometry and profiles, a concrete example of a shearless transport barrier that occurs in the core of a zonal \(E \times B\) jet using the global total-\(f\) gyrokinetic particle-in-cell code XGC.
We identify two channels, particle transport and turbulence spreading, in which the shearless transport barrier plays a role.
To show that this transport barrier corresponds to a shearless transport barrier, we extract the dominant global drift wave mode localized to the zonal jet, and use it to construct a model for the collisionless dynamics of gyrokinetic test particles experiencing perturbations from the wave.
We identify shearless tori in a Poincar\'{e} map constructed for the model and show that the tori lead to a significant reduction in particle transport across the shearless region.
Furthermore, we identify corresponding shearless phase-space transport barriers in the self-consistent collisional gyrokinetic Vlasov dynamics simulated by XGC to show that the shearless tori play a dynamically relevant role in the turbulence.
Compared to past works on shearless tori in plasma contexts, we emphasize the role played by trapped particle orbits.

We also consider the dynamics of the shearless phase-space transport barriers in fully developed turbulence beyond the Poincar\'{e} map model.
Using tools from topological data analysis (TDA), we demonstrate a correlation between heat flux avalanche events observed in the turbulence along with the radial propagation of phase space hole/blob features in the XGC simulations.
Upon reaching the shearless region, we show that these blobs cause eddy detachment events associated with the zonal jet, facilitating transport across the barrier without entirely destroying it.
We point out an analogy between this process and the formation of ``warm core ring'' and ``cold core ring'' structures in oceanic jets \citep{the_ring_group_gulf_1981,olson_rings_1991}.

The paper is organized as follows:
In \S\ref{sec:phenomenology}, we identify several phenomenological aspects of turbulence in the XGC simulations correlated with regions where the \(E \times B\) shear crosses through zero, suggesting the presence of a shearless transport barrier.
Emphasis is placed on quantities which can also be observed experimentally.
Next, in \S\ref{sec:shearless_theory}, we describe the construction of a single-mode test particle map model from the simulation data.
We analyze this map model using tools from dynamical systems, giving conditions necessary for the existence of shearless tori.
In \S\ref{sec:shearless_application} we apply this theory to fluctuations extracted from the XGC simulations, and provide direct evidence that structures associated with shearless tori are present in the XGC simulations.
In \S\ref{sec:discussion} we discuss the applicability of test particle map models more broadly, and suggest other turbulence regimes where similar test particle map models might reveal the existence of shearless transport barriers.
Finally, in \S\ref{sec:summary} we summarize the paper and outline possible future work to characterize the impact of shearless transport barriers on experiments and reactor design.

\section{Properties of Zonal Jets Observed in Gyrokinetic Simulations}\label{sec:phenomenology}

We begin in this section by introducing the high-fidelity flux-driven global gyrokinetic simulations which are the main subject of this work.
In \S\ref{subsec:xgc} we describe the parameters of the simulation and lay out the basic coordinate conventions used throughout the work.
Then in \S\ref{subsec:shearless_simulation} we discuss in detail the macroscopic properties of the transport and turbulence associated with a persistent shearless region observed in the simulations.

\subsection{XGC Simulations} \label{subsec:xgc}

This work considers turbulence simulations carried out using the global total-$f$ gyrokinetic particle-in-cell code XGC1 \citep{ku_xgc_2018,ku_fast_2018,hager_electromagnetic_2022}.
The results have been originally reported in \citet{zhu_intrinsic_2024}, where the simulation data is also provided \citep{zhu_intrinsic_2024_dataset}.
The simulations primarily model the interactions between ion temperature gradient (ITG) turbulence and zonal flows during an early ELM-free H-mode using realistic DIII-D tokamak geometry.
The equilibrium profiles are adapted from DIII-D shot number 141451 \citep{muller_experimental_2011,muller_intrinsic_2011}.

The simulations are electrostatic, and we denote the electrostatic potential by \(\phi\).
Deuterium ions and drift-kinetic electrons are simulated.
Their equilibrium density and temperature profiles are shown in figure~\ref{fig:poloidal_jet}(a,b), and their distribution functions $F_s$ evolve via the gyrokinetic Vlasov equation
\begin{equation}
\label{eq:XGC_vlasov}
d_t F_s=\partial_t F_s+\dot{\boldsymbol{R}}\cdot\nabla F_s+\dot{p}_\parallel \partial_{p_\parallel}F_s=C_s+S_s+N_s,
\end{equation}
where $s=i,e$.
The phase space advection operator \(d_t\) depends on the gyrokinetic particle equations of motion $(\dot{\boldsymbol{R}},\dot{p}_\parallel)$, given later in \eqref{eq:gk_char}.

On the right-hand side, $C_s$ is a multi-species Fokker-Planck-Landau collision operator \citep{yoon_fokker-planck-landau_2014,hager_fully_2016}, $S_s$ describes heating, and $N_s$ describes neutral ionization and charge exchange \citep{ku_fast_2018}.
In this simulation, turbulence is dominated by ion dynamics driven by the ion temperature gradient, and a 1MW heating is applied to ions in the core to sustain the temperature gradient.
Neutral dynamics are also included in the edge and scrape-off layer, providing the only particle source in the simulations, as there was no significant beam fueling.
The simulation consists of 16 poloidal planes that span \(1/3\) of the torus, with approximately 132k mesh nodes per poloidal plane.
We refer the reader to \cite{zhu_intrinsic_2024} for more details of the simulation.

\begin{figure}
	\centering
	\begin{tabular}{c}
		\includegraphics{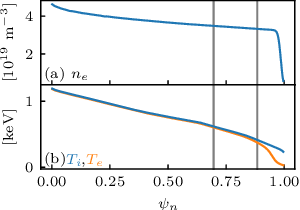} \\
		\includegraphics{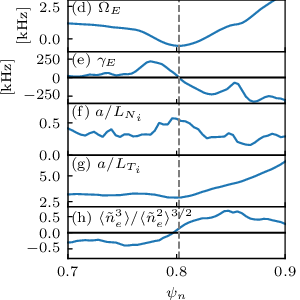}
	\end{tabular}%
	\begin{tabular}{c}
		\includegraphics{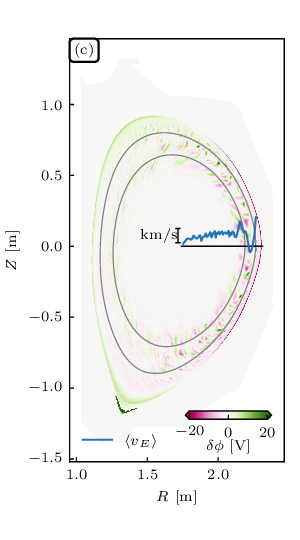}	
	\end{tabular}
	\caption{
		(a-b) Equilibrium profiles used to initialize the simulations.
		The region of interest is marked between the vertical gray lines.
		(c) A poloidal cross-section showing the non-zonal component of the electrostatic potential \(\delta \phi\), as well as the \(E \times B\) velocity \(v_{E}\) evaluated at the outboard midplane. The region of interest is between the flux surfaces indicated by gray contours.
		(d-h) Temporally-averaged quantities after several turbulence times as a function of radial coordinate, including the zonal \(E \times B\) rotation rate \(\Omega_E\), the Waltz-Miller shearing rate \(\gamma_E\), ion gyrocenter density gradient scale length \(a/L_{N_i}\), ion temperature gradient \(a/L_{T_i}\), and density skewness \(\langle \tilde{n}_e^3 \rangle / \langle \tilde{n}_e^2 \rangle^{3/2}\).
		The shearless region \(\psi=\psi^*\) is marked with a dashed vertical line.
		}
	\label{fig:poloidal_jet}
\end{figure}

For coordinate conventions, we use cylindrical coordinates consisting of the major radius \(R\), vertical coordinate \(Z\), and toroidal angle \(\varphi\).
The sign convention for \(\varphi\) is chosen such that \((R,\varphi,Z)\) is a right-handed triple.
We primarily focus on closed field line regions, and use \(\langle \cdot \rangle\) to denote the flux surface average.
The normalized poloidal flux \(\psi_n := \psi/\psi_{lcfs}\), defined in terms of the poloidal flux \(\psi\) and its value \(\psi_{lcfs}\) at the last closed flux surface (LCFS), can be used as a radial coordinate in this case.

We will also occasionally use the radial coordinate \(\rho = \rho_{tor} := \sqrt{\Phi_n}\), where \(\Phi_n = \Phi / \Phi_{lcfs}\) is the normalized toroidal flux.
These two radial coordinates are related by
\begin{equation*}
	\frac{\mathrm{d}\rho_{tor}}{\mathrm{d}\psi_n} = \frac{q}{2\rho_{tor}} \frac{\psi_{lcfs}}{\Phi_{lcfs}}
\end{equation*}
where \(q\) is the safety factor.
In the given coordinate conventions, the toroidal components the magnetic field and plasma current are both aligned with the \(-\varphi\) direction, so \(q\) and \(\Phi_{lcfs}\) are negative.

Intuitively, a zonal jet can be understood as a distinct maximum or minimum of the zonal flow.
One measure of this is to consider the poloidal component of the zonally-averaged \(E \times B\) flow \(\langle v_E \rangle := \partial_R \langle \phi \rangle / B\) evaluated on the outboard midplane.
Positive radial electric field \(E_r > 0\) will correspond to clockwise flows around the magnetic axis in the \((R,Z)\) plane, which are \(-Z\) directed on the outboard midplane.
Additionally, we also consider the toroidal rotation rate \(\Omega_E(\psi)\) associated with the radial electric field
\begin{equation*}
	\Omega_E(\psi) = \partial_\psi \langle \phi \rangle
\end{equation*}
In this case \(E_r > 0\) corresponds to \(\Omega_E < 0\), which is toroidal flow directed in the \(-\varphi\) direction.
Note that the ion \(\nabla B\) drift is pointed downward towards the X-point, and the ion diamagnetic drift is clockwise.

In Figure~\ref{fig:poloidal_jet}(c), we show a poloidal cross-section of the non-zonal component of the electrostatic potential as well as \(\langle v_E \rangle\) overplotted.
Notice that in the region inside the first gray contour, which we refer to as the `inner core', the turbulence is weak, and there is not a distinct radial structure to the sheared zonal flows.
Meanwhile, in the region between the first and second gray contours, which we refer to as the `outer core', the sheared zonal flows become significantly stronger.
This region is also known as ``no man's land'' (NML) in the literature, as it connects the pedestal to the core.
Moreover, there is a distinct banding of electrostatic potential fluctuations associated with the width of the zonal jet.
Outside of the second gray contour but inside the last closed flux surface, in the region which we will refer to as the `edge', there is a strong shear layer corresponding to the strong sheared flows present in the pedestal of the H-mode plasma.
The pedestal is contained within the edge.
In the following analyses, we focus on the outer core region.

\subsection{Properties of the Shearless Region} \label{subsec:shearless_simulation}

We now proceed to describe macroscopically observable properties of the shearless region associated with the zonal jet.
In figure~\ref{fig:poloidal_jet}(d-h), we show temporally-averaged quantities of various quantities as a function of normalized poloidal flux \(\psi_n\).
Panel (d) shows the zonal rotation rate \(\Omega_E\), showing the robust jet structure around \(\psi_n \approx 0.8\).
In panel (e) we plot the Waltz-Miller \(E \times B\) shear parameter \citep{waltz_theory_1998}
\begin{equation*}
    \gamma_E := (\rho/q) \partial_\rho \Omega_E \propto \partial_\psi^2 \langle \phi \rangle
\end{equation*}
and we define the non-degenerate shearless region as the region where \(\gamma_E\) passes through zero.
Note that this shear parameter is a flux surface averaged quantity, in contrast to the Hahm-Burrell \(E \times B\) shearing rate, see for example the discussion in \citet{burrell_role_2020}.
We will discuss how this choice of shear parameter arises from dynamical systems considerations in section \ref{sec:shearless_theory}.

The shearless region appears to be correlated with transport behavior of \(a/L_{N_i}\) and \(a/L_{T_i}\), which are the ion gyrocenter density and ion temperature gradient scale lengths respectively.
Panel (f) shows there is a strong enhancement of the ion gyrocenter (GC) density gradient at the shearless region.
Noting that the ion GC density can be identified with the potential vorticity \citep{mcdevitt_poloidal_2010}, this suggests the jet is supported by a so-called ``potential vorticity front''.
The ion GC density can be related to the usual electron density \(n_e\) and the electrostatic potential via the long-wavelength limit of the gyrokinetic Poisson equation,
\begin{equation*}
	N_i = n_e - \nabla_\perp \cdot \left(\frac{n_{i0} m_i}{Z_i e B^2} \nabla_\perp \phi\right)
\end{equation*}
where \(\nabla_\perp\) is the component of the gradient perpendicular to \(\mathbf{B}\), \(Z_i e\) is the ion charge, and \(n_{i0}\) is the equilibrium ion density.

From panel (g), the shearless region also appears to separate the region of flat \(a/L_{T_i}\) in the inner core from the region of steeper \(a/L_{T_i}\) in the edge.
This is reminiscent of the argument in \citet{singh_when_2020} which argues that turbulence spreading from the edge contributes to the weakening of temperature stiffness in NML.
To quantify the interaction between the shearless region and turbulence propagation, we consider the (normalized) density skewness \(\langle \tilde{n}_e^3\rangle / \langle \tilde{n}_e^2 \rangle^{3/2}\).
Skewness measures the asymmetry of a probability distribution about its mean, with positive values indicating an excess of positive fluctuations and a negative value indicating an excess of negative fluctuations.
When these fluctuations are carried by isolated structures, they are often referred to as density ``blobs'' or ``holes'' respectively.
Density skewness is frequently used as an indicator of blob/hole and avalanche dynamics in both experimental and theoretical studies \citep{dippolito_convective_2011}.
We will show later in \S\ref{subsec:eddy_detachment} that these statistical fluctuations are associated with filamentary phase-space structures, so we will also refer to them as blobs/holes.
Panel (h) shows that the zero of the density skewness, associated with a transition from hole- to blob-dominated fluctuations, occurs at the same radial location as the shearless region.
This suggests that the shearless region acts as the boundary between the inner core and the edge, preventing turbulence spreading from one region to the other.

\begin{figure}
	\centering
	\includegraphics{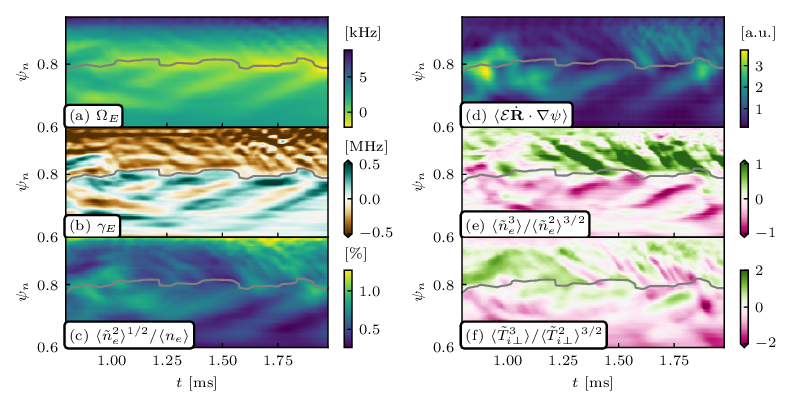}
	\caption{Sequence of color plots showing the evolution of various flux-surface-averaged quantities over time and radial coordinate. The shearless region is overplotted with a gray line on all of the plots.}
	\label{fig:hovmoller}
\end{figure}

Moving to the time-dependent behavior of the zonal jet, in Figure~\ref{fig:hovmoller}, we show a sequence of color plots showing the evolution of various flux-surface-averaged quantities over the time and radial coordinates.
Panels (a,b) illustrate the zonal jet and the zonal shear parameter \(\gamma_E\).
There is a persistent region of zero shear, equivalently a local minimum in \(\Omega_E\), in the vicinity of the zonal jet near \(\psi_n \approx 0.8\).
Outside of the zonal jet, the zonal flows display more fluctuating behavior.
The space-time plots suggest intermittent increases in zonal flow strength that begin in the inner core and edge which then converge radially to the zonal jet.

These radially converging zonal flow bursts can be correlated with the radial propagation of turbulence avalanches.
In panel (c) we show the normalized flux-surface-averaged density fluctuations \(\langle\tilde{n}_e^2\rangle^{1/2}/\langle n_0 \rangle\).
The density fluctuations also display this radial convergence behavior, with bursts of fluctuation amplitude that form in the inner core and the edge and converge towards the zonal jet.
This behavior is also reflected in the radial energy flux \(\langle \mathcal{E} \dot{\mathbf{R}} \cdot \nabla \psi \rangle\), shown in panel (d), where \(\mathcal{E}\) is the particle kinetic energy and \(\dot{\mathbf{R}}\) is the particle drift.
Panel (e) shows the density skewness, which again shows this radial converging behavior indicative of blobs/holes forming then propagating towards the zonal jet.
Similar behavior is observed in the ion temperature skewness \(\langle \tilde{T}_i^3\rangle / \langle \tilde{T}_i^2 \rangle^{3/2}\), which is shown in panel (f).

Note the anti-correlation between the density and temperature skewness near the shearless region seen in panels (e,f) show that the density blobs (resp. holes) correspond to temperature holes (resp. blobs).
The temperature holes (resp. blobs) propagate inward (resp. outward), which matches the expectation for ion temperature gradient (ITG) turbulence, where the ion temperature serves as the primary source of free energy, whereas density gradients are stabilizing.
This differs from resistive drift wave and interchange-like turbulence, where density gradients are destabilizing and thus act a source of free energy for the turbulence, leading to the usual picture of inward propagating holes and outward propagating blobs in edge turbulence.
Since electron density fluctuations are more easily measured than ion temperature fluctuations, we will use the electron density fluctuations as a proxy for the ion temperature fluctuations in the following analyses, and use the terminology for blobs and holes accordingly.

In summary, both the time-averaged and time-dependent analyses suggest the presence of a robust shearless region associated with a zonal jet near \(\psi_n \approx 0.8\).
The zonal jet is linked to a significant increase in the ion gyrocenter density gradient and appears to act as a barrier to turbulence propagation, potentially suggestive of a transport barrier associated with the shearless region.
In the following sections, we will demonstrate how these features could be explained as a result of a shearless transport barrier associated with the zonal jet.

\section{Theory of Shearless Tori in Gyrokinetic Turbulence}\label{sec:shearless_theory}

In this section, we develop the test particle map model which we will use to study shearless transport barriers in the gyrokinetic simulations.
We begin in \S\ref{subsec:model_construction} by constructing a model for the electrostatic perturbations observed in the gyrokinetic simulations in the vicinity of the shearless region.
We follow this in \S\ref{subsec:planar_map} where we show that the gyrokinetic test particle dynamics in the presence of the model perturbation have an exact reduction to a planar map, which allows for the identification of shearless invariant tori associated with non-degenerate maxima and minima of the rotation number.
In \S\ref{subsec:invariant_tori} we show that these shearless invariant tori persist in the presence of the model perturbation and act to reduce the level of particle transport across the shearless region.

\subsection{Construction of the Model Dynamical System} \label{subsec:model_construction}

As reviewed in the introduction, there has been extensive work studying the role of shearless tori in magnetically confined plasmas.
In this section, we will summarize the key parts of the theory of shearless tori relevant to this work, as well as discuss the requirements on the drift wave fluctuations necessary for the theory to apply.
The key idea is that for fluctuations which locally ``look like'' a rigidly toroidally rotating perturbation, the presence of an additional gyrokinetic invariant which allows for the exact reduction of the gyrokinetic test particle dynamics to a planar map.
We argue that for marginally unstable drift waves in the presence of a zonal jet, the phenomenon of wave trapping will tend to lead to just ``a few'' modes active at any given instant of time, leading to fluctuations which approximately satisfy this condition.

We consider the gyrokinetic characteristic equations used in XGC,
\begin{subequations} \label{eq:gk_char}
	\begin{gather}
		B_\parallel^* \dot{\mathbf{R}} = \frac{1}{Z_s e} \hat{\mathbf{b}} \times \nabla H + v_\parallel \mathbf{B}^* \\
		B_\parallel^* \dot{p}_\parallel = -\mathbf{B}^* \cdot \nabla H
	\end{gather}
\end{subequations}
where the overdot is the time derivative, giving the evolution of the gyrocenter \(\mathbf{R}\) and parallel momentum \(p_\parallel\) a given particle.
\(m_s\) and \(Z_s e\) is the species mass and charge respectively.
We focus on the electrostatic gyrokinetic Hamiltonian,
\begin{equation}
	H = \frac{p_\parallel^2}{2m_s} + \mu B + Z_s e \mathcal{J}[\phi]
\end{equation}
although the theory can also be developed for electromagnetic perturbations, see for example \citet{anastassiou_role_2024}.
Here \(v_\parallel = \partial p_\parallel H\) is the parallel velocity, \(\mu\) is the magnetic moment, \(\mathcal{J}\) is the gyro-average operator, \(\phi\) is the electrostatic potential, \(m\) and \(Z_s e\) are the species mass and charge, \(\hat{\mathbf{b}} = \mathbf{B}/B\), \(\mathbf{B}^* = \mathbf{B} + \nabla \times (p_\parallel \hat{\mathbf{b}} / Z_s e)\), and \(B_\parallel^* = \hat{\mathbf{b}}\cdot \mathbf{B}^*\).

For an axisymmetric system \(\partial_\varphi H = 0\), the canonical toroidal angular momentum
\begin{equation}
	P_\varphi = Z_s e \psi + p_\parallel \hat{\mathbf{b}} \cdot R^2 \nabla \varphi
\end{equation}
is conserved along gyrokinetic characteristics.
In addition, for time-independent systems \(\partial_t H = 0\), the Hamiltonian \(H\) will be conserved along characteristics.
These two facts lead to the complete integrability of gyrokinetic particle trajectories in axisymmetric time-independent systems, where particles enjoy three invariants of motion \((\mu, H, P_\varphi)\).

When considering perturbations of the fields with an axisymmetric time-independent background, the symmetries imply that eigenmodes of the system will have electrostatic potentials with the form
\begin{equation*}
	\delta \phi \sim e^{i (n \varphi - \omega t)} \delta \phi(R,Z).
\end{equation*}
Such modes correspond to rigidly toroidally rotating fluctuations, with a toroidal angular phase velocity of \(\Omega = \omega/n\).
Hamiltonians consisting of rigidly toroidally rotating modes with a common angular phase velocity \(\Omega\) will satisfy \(\partial_t H + \Omega \partial_\varphi H = 0\).
Particles undergoing motion in such fields will have two invariants of motion, \((\mu, K)\), where \(K=H - \Omega P_\varphi\) is the Hamiltonian in the rotating frame.
This fact is often utilized in studies of energetic particle transport \citep{hsu_alpha-particle_1992,todo_introduction_2019}.
Note that a regime where most fluctuations share a nearly common angular phase velocity of \(\Omega\) would correspond to weakly dispersive turbulence.

The key hypothesis we take is that in the presence of a zonal jet, the phenomenon of wave trapping will produce fluctuations which lead to this weakly dispersive regime.
In \citet{zhu_theory_2020,zhu_theory_jpp_2020}, it was demonstrated how wave trapping and anti-trapping play a key role determining the Dimits shift regime in fluid models of drift-wave turbulence.
The key physics of wave trapping is that waves can constructively interfere.
Physically, one can imagine that the \(E \times B\) jet acts as a ``cavity'' which modifies the structure of the primary ITG instability.
In the presence of shear of opposing signs, discrete quantum harmonic oscillator-like eigenmodes can appear which have radially-localized fluctuation envelopes.
Intuitively, this radial localization increases the effective radial wavenumber, typically leading to a stabilizing effect on radial-gradient-driven instabilities.
In such a regime turbulent mixing can become strongly spatially inhomogeneous; see for example \citet{cao_nearly_2023,cao_maintenance_2024} for detailed studies of the nonlinear dynamics of large-amplitude tertiary instabilities in a variant of the Hasegawa-Wakatani equations.

\begin{figure}
	\centering
	\includegraphics{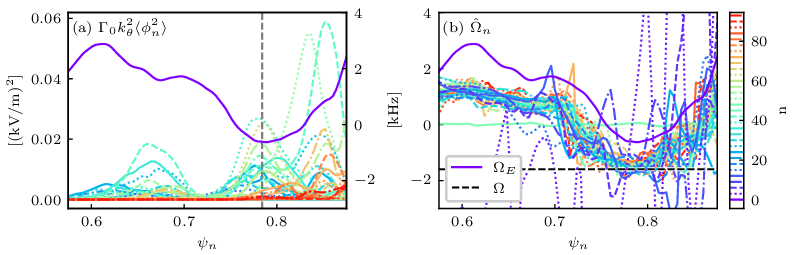}
	\caption{(a) Toroidal mode spectra and (b) angular phase velocities, both measured at one instant in time.
    \(\langle \phi_n^2 \rangle\) is the flux surface averaged squared electrostatic potential, and \(\hat{\Omega}_n\) is the toroidally-directed angular phase velocity.
    The quantity \(\Gamma_0 k_\theta^2 \langle \phi_n^2 \rangle\) mimics the gyroaveraged \(\langle\mathcal{J}[E_\perp]^2\rangle\) spectrum.
    Different toroidal mode numbers are shown with different colors and linestyles.
    \(\Omega_E\) is shown on both plots as a solid purple line for reference, and the location of the shearless region in the zonal jet is demarcated in panel (a) by a dashed vertical gray line.
    A dashed horizontal black line is also plotted in in panel (b) over the angular phase velocities to show the rotation rate \(\Omega\) used for the model fluctuations.}
	\label{fig:dpot_spectra_fit}
\end{figure}

To show that this hypothesis is reasonable in practice, we turn to the electrostatic fluctuations observed in the gyrokinetic simulation data.
In figure~\ref{fig:dpot_spectra_fit}, the toroidal mode spectra and angular phase velocities are plotted for one instant in time.
The toroidal mode spectra are computed by taking the Fourier representation
\begin{equation*}
	\phi(R,\varphi,Z) = \sum_{n=-\infty}^{\infty} e^{in\varphi} \phi_n(R,Z)
\end{equation*}
Note the data is upsampled from the original 16 toroidal planes to 48 toroidal planes using field-line following interpolation, which allows the usual fast Fourier transform algorithm to be used without aliasing issues due to the differing scales of \(k_\parallel\) and \(k_\perp\).
As a proxy for the role of gyroaveraging on the fluctuation spectrum, which we will explore in more detail later, we apply an effective gyroaverage factor \(\Gamma_0\).
Thus, we compute a proxy for the perpendicular electric field spectrum by multiplying the zonally averaged potential \(\langle \phi_n^2 \rangle\) by \(\Gamma_0(b) k_\theta^2\), where \(k_\theta := n q / \rho\) is the effective poloidal wavenumber, \(b = (k_\theta \rho_i)^2\) depends on the gyroradius \(\rho_i\) of a thermal ion at the outer midplane, and \(\Gamma_0(b) = I_0(b) e^{-b}\) is the effective gyroaverage factor defined in terms of the modified Bessel function of the first kind \(I_0\).
This mimics the gyro-averaged perpendicular electric field \(\langle (\nabla_\perp \mathcal{J}[\phi])^2 \rangle\).
To compute the toroidal (real) angular phase velocity, for each flux surface \(\psi\) we compute the complex scalar \(c_n(\psi)\) which minimizes
\begin{equation*}
    \langle (\phi_n(t=t) - c_n(\psi) \phi_n(t=t-\Delta t))^2 \rangle
\end{equation*}
for each flux surface \(\psi\).
We focus on a single time \(t \approx 1.584\) ms, when the fluctuations are strong.
If the electrostatic potential were to undergo purely rigid toroidal rotation with angular frequency \(\Omega\), then the amplitudes of the Fourier modes would evolve in time as \(\phi_n \sim e^{-i n \Omega t}\).
Taking \(c_n(\psi) = e^{-i\varpi_n(\psi) \Delta t}\) allows us to extract a dominant toroidally-directed angular phase velocity \(\hat{\Omega}_n(\psi) := \operatorname{Re}[\varpi_n(\psi)] / n\) on each flux surface \(\psi\) for each toroidal mode \(n\).

One key observation from Figure~\ref{fig:dpot_spectra_fit}(a) is that the fluctuations at a given radial location tend to be rather sparse in toroidal mode number.
Furthermore, rather than extending across the entirety of the plasma, fluctuations at a given toroidal mode number tend to have a radially localized envelope.
Several toroidal modes appear to have envelopes localized within the zonal jet, whose location is demarcated by a dashed vertical line.
Another key observation from Figure~\ref{fig:dpot_spectra_fit}(b) is that the phase velocities \(\hat{\Omega}_n\) have much less spread near the zonal jet as well.
These observations support the hypothesis of wave trapping leading to weak dispersion.
We compute an averaged rotation rate \(\Omega\), also shown in panel (b), over all the fluctuations to use in the following.

\begin{figure}
	\centering
	\includegraphics{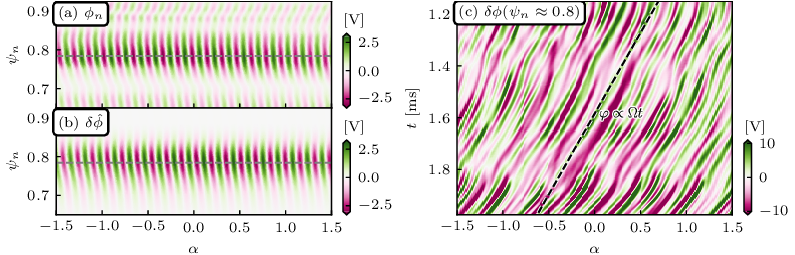}
	\caption{Poloidal slices of (a) The electrostatic potential \(\phi_n\) for the \(n=39\) Fourier mode, and (b) the single-mode model electrostatic potential \(\delta \hat{\phi}\) at a single instant in time. The fields are plotted against the radial coordinate \(\psi_n\) and the perpendicular field line label \(\alpha = \varphi - q \theta\). The shearless region \(\psi=\psi^*\) is demarcated with a dashed gray line. (c) Space-time plot showing the value of the electrostatic potential fluctuations \(\delta \phi\), evaluated at the shearless region \(\psi = \psi^*\) plotted against \(\alpha\). A line with \(\varphi \propto \Omega t\) is shown demonstrating the fixed phase velocity of the fluctuations.}
	\label{fig:dpot_hov_model}
\end{figure}

In order to take advantage of this wave trapping to create a model of the dynamics, we extract out a single mode with the largest amplitude in the zonal jet and model the electrostatic potential with a rigidly toroidally rotating version of this mode.
We remark that the key feature is the toroidally rigid rotation of the mode; additional modes could be added to the model with the assumption that they rotate at the same toroidal frequency, but in practice it was found that one mode was enough to reproduce key qualitative features of the test particle dynamics in the shearless region.

Since the mode is entirely localized within the last closed flux surface (LCFS), we can give an expression for the mode using straight field line coordinates \((\psi, \varphi, \theta)\) where \(\psi\) is the radial coordinate, \(\varphi\) is the usual toroidal angle, and \(\theta\) is the poloidal angle with \(\theta=0\) taken at the outboard midplane.
We also define the perpendicular field line label \(\alpha = \varphi - q(\psi) \theta\).
Note that the particle dynamics are evolved using the original cylindrical coordinates \((R,\varphi,Z)\), and interpolating functions \(\psi(R,Z)\) and \(\theta(R,Z) = \operatorname{arctan}(\tfrac{Z-Z_{axis}}{R-R_{axis}}) + \delta \theta (R,Z)\) are used to relate the cylindrical coordinates with the straight field line coordinates.

We take a model electrostatic potential \(\hat{\phi}\) with zonal and non-zonal parts \(\hat{\phi} = \langle\phi\rangle(\psi) + \delta\hat{\phi}(\psi,\varphi,\theta,t)\).
For the zonal part, we directly take the zonally-averaged electrostatic potential from the simulation data.
For the non-zonal part, we use the ballooning representation to model a ballooning mode with radial envelope
\begin{equation} \label{eq:ballooning_mode}
	\delta \hat{\phi}(\psi, \varphi, \theta, t) = \operatorname{Re} [e^{in\varphi'} \hat{\phi}_n(\psi, \theta)] = \operatorname{Re}\sum_{\ell=-1}^{1} e^{in(\varphi' - q(\psi) (\theta - \theta_0 + 2 \pi \ell))} g_n(q(\psi), \theta + 2 \pi \ell)
\end{equation}
where \(\varphi' = \varphi - \Omega t\) is the toroidal angle in a rotating frame using the averaged angular phase velocity \(\Omega\) computed earlier.
In principle ballooning modes include an infinite sum over \(\ell\) from \(-\infty\) to \(\infty\), but in practice the summation bounds \(\pm 1\) were enough to model the mode.
Ballooning modes are characteristic for drift wave turbulence, since they can localize to the bad curvature region in the outboard midplane while minimizing \(k_\parallel\) in the presence of magnetic shear.

The envelope function \(g_n(q,\eta)\) is taken as a sum of Gauss-Hermite functions in \(q(\psi)\) and the field line coordinate \(\eta\),
\begin{equation*}
	g_n(q,\eta) = \sum_{j+k\le 2} a_{j,k} He_j(z_q) He_k(z_\eta) e^{-(z_q^2 + z_\eta^2)/2}
\end{equation*}
where \(z_q = (q-q_0)/\sigma_q, z_\eta = (\eta-\theta_0)/\sigma_\eta\) are the normalized radial and field-line coordinates.
The parameters \(q_0, \theta_0, \sigma_q, \sigma_\eta\) are all taken as real, while \(a_{j,k}\) are complex.
The parameter values were found using the \texttt{scipy.optimize.minimize} routine to minimize the objective function
\begin{equation*}
    \int_{\psi_0}^{\psi_1} \langle |\phi_n - \hat{\phi}_n|^2 \rangle \operatorname{d}\psi
\end{equation*}
for \(n=39\), which is mean squared deviation of the model \(\hat{\phi}_n\) from the the actual electrostatic potential \(\phi_n\) for the dominant toroidal mode number at this instant of time in the zonal jet.

We compare the model electrostatic potential against the electrostatic fluctuations from the simulations in figure~\ref{fig:dpot_hov_model}.
In panels (a,b) we show a comparison between the electrostatic potential \(\phi_n\) of the \(n=39\) Fourier mode from the simulations against the model electrostatic potential \(\delta \hat{\phi}\).
The dominant radial band in \(\phi_n\), localized near \(\psi=\psi^*\), is well-approximated by the model fluctuation \(\delta\hat{\phi}\).
Furthermore in panel (c), we compare the hypothesis of a constant-speed rotating frame against the phase speed of the electrostatic fluctuations at the shearless region \(\psi = \psi^*\).
The slope of the line \(\varphi' = \varphi - \Omega t = const.\) matches the slope of the electrostatic fluctuations, showing the applicability of this hypothesis.
Finally, we remark that while the conservation of \(K\) arising from the weak dispersion hypothesis significantly simplifies the analysis of the test particle dynamics, recent advances in non-twist KAM theory \citep{gonzalez-enriquez_singularity_2014} suggest this assumption may be much stronger than necessary, which we will discuss in the next section and later in \S\ref{subsec:assumptions}.

\subsection{Planar Map Dynamics and Conditions for Shearless Tori} \label{subsec:planar_map}

In this section, we describe how the model dynamical system, which is constructed in the full gyrokinetic phase space, can be reduced to a planar map through an appropriately chosen surface of section, i.e. a Poincar\'{e} section.
This process is exactly analogous to the usage of surfaces of section to study energetic particle transport under the influence of a single toroidal Alfv\'{e}n eigenmode \citep{hsu_alpha-particle_1992,todo_introduction_2019}.
Similar to how the safety factor \(q\) captures topological information about magnetic field lines through their average winding numbers around the torus, we use the kinetic safety factor \(q_{kin}\) \citep{gobbin_resonance_2008} to infer topological information about the test-particle orbits.
The presence of non-degenerate minima and maxima in \(q_{kin}\) provides a necessary condition for the potential existence of phase-space transport barriers in the form of shearless invariant tori.

We begin by reviewing how gyrokinetic dynamics for the single-mode model can be reduced to a planar map through an appropriately chosen surface of section.
This analysis is most easily performed using cylindrical coordinates \((R,\varphi,Z)\) for the spatial variables.
Since the model Hamiltonian is time-independent in a frame rotating with angular velocity \(\Omega\), we can use the moving coordinate \(\varphi' = \varphi - \Omega t\) where the fluctuations will not have any explicit time dependence.
Thus, \(t\) is an ignorable coordinate, and \((\mu, K)\) will be constants of motion.
For a given set of invariants \(\mu,K\) in the rotating frame, we define the manifolds
\begin{gather*}
	E_{\mu,K} := \{(R,\varphi',Z,p_\parallel,\mu) : \mu=\mu, H-\Omega P_\varphi = K\} \\
	\Gamma^{\pm}_{\mu,K} := \{z \in E_{\mu,K} : Z = Z_{axis}, R>R_{axis}, \pm\dot{Z} > 0\}
\end{gather*}
\(E_{\mu,K}\) is the set of all particle configurations with a given \(\mu,K\).
It is defined by two constraints in the 5-dimensional particle phase space, so \(E_{\mu,K}\) is a 3-dimensional manifold.
\(\Gamma^{\pm}_{\mu,K}\) is the intersection of \(E_{\mu,K}\) with the outboard midplane, selecting only particle configurations which pass through with either positive or negative velocity.
\(\Gamma^{\pm}_{\mu,K}\) will be a 2-dimensional manifold, and it can be naturally parameterized by the two coordinates \(P_{\varphi}, \varphi'\).
Since \(P_\varphi = Ze\psi + O(\rho_{pol}/a)\), \(P_\varphi\) can be thought of as a radial coordinate.

Notice that \(E_{\mu,K}\) is an invariant manifold of the gyrokinetic characteristic equations \eqref{eq:gk_char} for the model dynamical system, meaning that any particle which starts in \(E_{\mu,K}\) will remain in \(E_{\mu,K}\) for all time.
Furthermore, notice that passing or trapped trajectories within \(E_{\mu,K}\) will repeatedly pass transversally through at least one of the 2-dimensional sections \(\Gamma^{+}_{\mu,K}\) or \(\Gamma^{-}_{\mu,K}\).
This allows the definition of a Poincar\'{e} map \(f^{+}_{\mu,K}: \Gamma^{+}_{\mu,K} \to \Gamma^{+}_{\mu,K}\) (and similar for \(-\)) such that \(f^{+}_{\mu,K}(z)\) is the first return of the trajectory of a particle starting at \(z \in \Gamma^{+}_{\mu,K}\) back to \(\Gamma^{+}_{\mu,K}\).
In other words, we follow particles with a given \(\mu,K\) initialized on the outboard midplane, and compute their repeated intersections with the outboard midplane.
Exactly analogous to the techniques used in energetic particle transport, in the weakly dispersive regime we can use the planar maps \(f^{\pm}_{\mu,K}\) to study the exact dynamics of the gyrokinetic characteristic equations \eqref{eq:gk_char}.

Important topological information about particle orbits can be deduced from these planar maps using the kinetic safety factor \(q_{kin}\) \citep{gobbin_resonance_2008,anastassiou_role_2024}.
Similar to how the \(q\) profile gives the ratio of the average number of toroidal to poloidal transits made by a magnetic field line, \(q_{kin}\) gives the ratio of the average number of toroidal to poloidal transits made by a particle orbit.
In the integrable case, the particles undergo quasiperiodic motion, and \(q_{kin}\) can be computed
\begin{equation} \label{eq:qkin_integrable}
    q_{kin} = \frac{\Omega_{\varphi'}}{\Omega_{\theta}}
\end{equation}
which is the ratio of the average toroidal transit frequency \(\Omega_{\varphi'}\) in the moving frame, i.e. the frequency by which the particle completes a \(2\pi\) orbit toroidally in the frame rotating with frequency \(\Omega\), to average poloidal return frequency \(\Omega_{\theta}\), i.e. the frequency by which the particle takes to return to the midplane.
More generally, \(q_{kin}\) is equal to the rotation number of an orbit under the action of the Poincar\'{e} map \(f^{\pm}_{\mu,K}\), relevant to dynamical systems theory.

\(q_{kin}\) has a different physical origin for passing and trapped particles due to the topological difference between passing and trapped orbits.
For passing particles, the toroidal and poloidal transit times are determined by the parallel transit times of the particle along magnetic field lines as they wind around the torus.
Meanwhile for trapped particles, the toroidal transit time is determined by the toroidal precession frequency of the bounce orbits, and the poloidal transit time is given by the bounce time.
In both the integrable and non-integrable cases, \(q_{kin}\) can be computed by integrating the characteristic equations \eqref{eq:gk_char} using an RK4 method for ODEs.
Whenever possible for the numerical integration, quantities involving derivatives of fields are expressed analytically in terms of derivatives of interpolating functions of the equilibrium poloidal flux function \(\psi(R,Z)\) and flux-surface functions \(f(\psi)\).
In the integrable case, the toroidal and poloidal frequencies associated with the orbit can be directly calculated from the trajectories to evaluate \eqref{eq:qkin_integrable}.
In the non-integrable case, \(q_{kin}\) can be computed using the weighted Birkhoff average \citep{sander_birkhoff_2020} to approximate the average rotation number of a particle trajectory under the Poincar\'{e} map.

To analyze the model dynamical system using \(q_{kin}\), we first consider the case where only the time-independent zonal component of the electrostatic potential \(\hat{\phi} = \langle \phi \rangle(\psi)\) is included.
In this case, the test-particle Hamiltonian is also axisymmetric, leading to conservation of the canonical toroidal angular momentum \(P_\varphi\).
This implies integrability for the gyrokinetic test-particle orbits, meaning that almost all particle orbits will be confined to a set of nested invariant tori which foliate the phase space.
These invariant tori will also be invariant tori of the planar maps \(f^{\pm}_{\mu,K}\), in which case they can be labeled by the radial coordinate \(P_{\varphi}\).

Similar to how all magnetic field lines in a given flux surface labeled by \(\psi\) will have the same safety factor \(q=q(\psi)\), all particles in a given invariant torus of the map \(f^{\pm}_{\mu,K}\) labeled by \(P_{\varphi}\) will have the same kinetic safety factor \(q_{kin}=q_{kin}(P_{\varphi})\).
In figure~\ref{fig:q_kinetic}, we show plots of \(q_{kin}(P_\varphi)\) for two different values of \(\mu, K\) corresponding to thermal passing and trapped particles.
Specifically, we initialize particles on the outboard midplane with kinetic energy equal to the equilibrium ion temperature at the shearless region \(\mathcal{E} = T_{i0}(\psi_n = 0.8)\).
For trapped particles, we use a ratio of perpendicular to total kinetic energy \(\mathcal{E}_\perp/\mathcal{E}=2/3\) at the outboard midplane, while for passing particles we take \(\mathcal{E}_\perp/\mathcal{E}=1/3\).
Passing particles complete orbits approximately follow the field lines, and have \(q_{kin}\) profiles similar to the \(q\) profile of the magnetic field lines.
Meanwhile, trapped particle orbits differ topologically from magnetic field lines as they do not wind around the torus.
The toroidal precession from \(E \times B\) rotation \(\Omega_E\) is the dominant factor determining \(q_{kin}\) in this case.
Note there is a nearly uniform offset of \(q_{kin}\) from \(\Omega_E\) in the trapped particle case, which likely results from the bounce-averaged ion \(\nabla B\)-drift.

\begin{figure}
    \centering
    \includegraphics{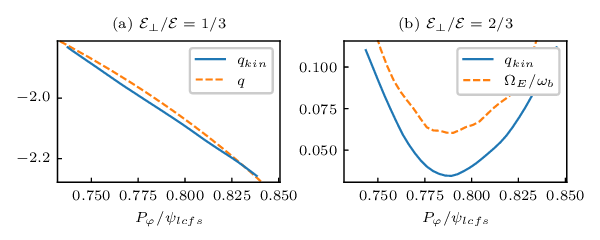}
    \caption{Plots of \(q_{kin}(P_\varphi)\) for (a) passing \(\mathcal{E}_\perp/\mathcal{E} = 1/3\) and (b) trapped particles \(\mathcal{E}_\perp/\mathcal{E} = 2/3\) in the vicinity of the zonal jet, with the kinetic energy \(\mathcal{E}\) equal to the ion temperature in the shearless region. The ratio of perpendicular to total kinetic energy \(\mathcal{E}_\perp/\mathcal{E}\) is taken at the outboard midplane. The magnetic safety factor \(q(\psi)\) and \(E \times B\) rotation \(\Omega_E(\psi)\) divided by the bounce frequency \(\omega_b\) are shown for comparison.}
    \label{fig:q_kinetic}
\end{figure}

For perturbations of the form given by \eqref{eq:ballooning_mode}, KAM theory can be invoked to argue that many of these invariant tori will survive perturbation.
Depending on \(q_{kin}'(P_{\varphi})\), which determines the \textit{twist} or \textit{shear} of the map \(f^{\pm}_{\mu,K}\), two variants of KAM theory are potentially applicable.
If \(q_{kin}\) is smooth and \(q_{kin}'(P_{\varphi}) \neq 0\) across the region of interest, the map \(f^{\pm}_{\mu,K}\) is known as a \textit{twist map}.
In this case, classical KAM theory is applicable, and it is typically expected that invariant tori with the most irrational rotation numbers \(q_{kin}\) will be the most robust under perturbation \cite{escande1985}.
Referring to figure~\ref{fig:q_kinetic}(a), this version of KAM would be relevant to the passing particles as the magnetic shear, which strongly dominates over the \(E \times B\) shear in determining the behavior of \(q_{kin}\), does not change sign.

A second case occurs when \(q_{kin}\) has a non-degenerate minimum or maximum at an isolated radial coordinate \(P_\varphi^*\), equivalently \(q_{kin}'(P_{\varphi}^*) = 0\) and \(q_{kin}''(P_{\varphi}^*) \neq 0\).
In this case \(f^{\pm}_{\mu,K}\) is known as a \textit{nontwist map}, and the invariant torus labeled by \(P_{\varphi}^*\) is known as a \textit{nontwist} or \textit{shearless torus}.
In Appendix~\ref{app:nontwist_tori}, we note that the definition of shearless torus is independent of the choice of Poincar\'{e} section, implying that the nontwist condition is a topological condition applicable to the continuous-time dynamics as well.
Despite classical KAM theory being inapplicable, shearless tori are phenomenologically observed to be extremely robust to perturbation, often times being the last remaining invariant torus in otherwise chaotic systems.
Physics-focused reviews of shearless tori can be found in \citet{del-castillo-negrete_chaotic_2000,morrison_magnetic_2000,caldas_shearless_2012}, and rigorous mathematical justification of their existence has been established with non-twist KAM theorems \citep{delshams_kam_2000}.

Referring to figure~\ref{fig:q_kinetic}(b), nontwist KAM theory would be relevant to the trapped particles.
The bounce motion averages out the effect of the magnetic shear, and \(E \times B\) shear from \(\Omega_E\) becomes the dominant factor influencing \(q_{kin}\).
Thus, the shearless region given by \(\partial_\psi \Omega_E(\psi) \propto \gamma_E = 0\) is potentially associated with shearless invariant tori in the trapped regions of phase space.
We remark that while we discussed these conditions for electrostatic transport, the nontwist condition is a condition only on the rotation numbers of particle orbits, and is applicable to any collisionless particle dynamics that can be described using a Hamiltonian formalism.
For example, in \citet{anastassiou_role_2024} the role of shearless transport barriers was discussed for magnetic perturbations in the edge.

We end this section by remarking that recent work has used singularity theory to classify shearless tori in higher dimension and prove corresponding nontwist KAM theorems about their persistence under perturbation \citep{gonzalez-enriquez_singularity_2014}.
For this work, we primarily focus on planar maps, where much more is known about the phenomenology of shearless tori.
However, as pointed out in earlier works \citet{del-castillo-negrete_area_1996,rypina_robust_2007}, one way to define shearless tori in integrable systems is through degeneracies in the quadratic part of the Hamiltonian.
In higher dimension, these degeneracies can be detected in action-angle coordinates by looking for values of \(J\) where the local shear \(D\omega(J)\) is singular.
Here \(D\omega(J)\) is the matrix of partial derivatives \([D\omega(J)]_{ij} = \partial_{J_j} \omega_i\) of the torus frequencies \(\omega_i = \partial_{J_i} H\) with respect to the action coordinates \(J_j\).
Observe that \(D\omega(J)\) is equal to the matrix of second derivatives of the Hamiltonian with respect to \(J\), so a shearless torus with action \(J^*\) can be associated with a degenerate critical point of an associated potential function \(V(J) = H(J) - \omega(J^*) \cdot J\), furnishing the link with singularity theory.
We will discuss some conjectures relating to these higher-dimensional shearless tori in \S\ref{subsec:assumptions}.

\subsection{Persistence of Shearless Tori in Model Dynamical System} \label{subsec:invariant_tori}

\begin{figure}
    \centering
    \includegraphics{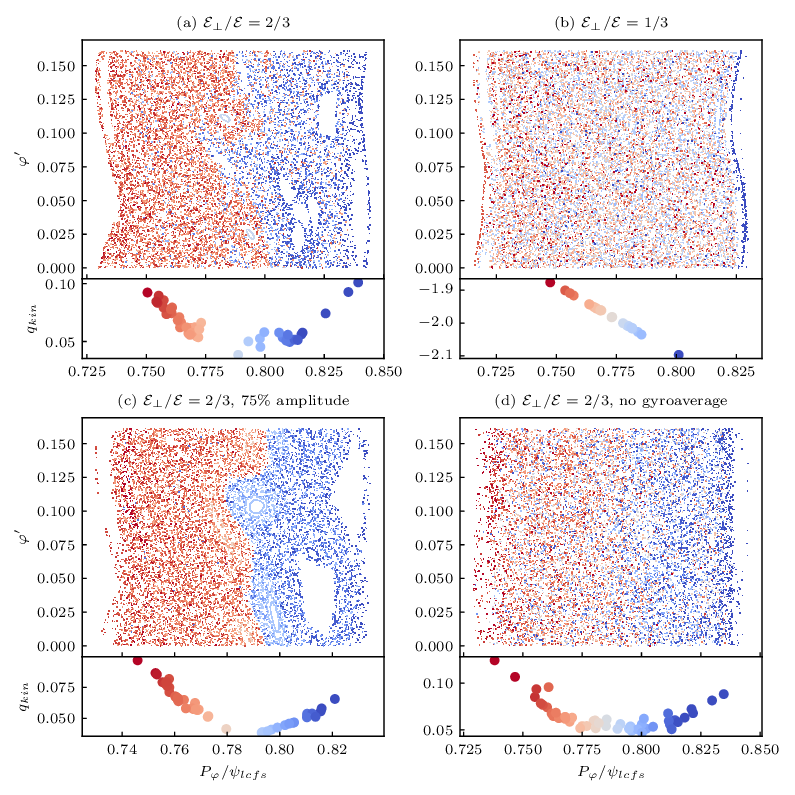}
    \caption{Poincar\'{e} sections of model gyrokinetic system with corresponding rotation numbers in the associated sub-panels below. Particle trajectories are colored by their average radial location, measured by \(P_{\varphi}\). The first row shows (a) trapped particles and (b) passing particles experiencing the drift wave at observed amplitude. The second row shows trapped particles experiencing (c) the drift wave at 75\% amplitude and (d) the drift wave at observed amplitude but without gyro-averaging applied}
    \label{fig:poincare_section}
\end{figure}

Now, we consider the full model dynamical system with electrostatic potential \(\hat{\phi} = \langle \phi \rangle + \delta \hat{\phi}\) including the ballooning mode \eqref{eq:ballooning_mode} extracted from the gyrokinetic simulation.
We consider the motion of the deuterium ions, as the modes rotate in the ion diamagnetic direction and interact more strongly with ion trajectories.
We take \(\mu, K\) to be the same as the passing \(\mathcal{E}_\perp/\mathcal{E}=1/3\) and trapped \(\mathcal{E}_\perp/\mathcal{E}=2/3\) cases considered in figure~\ref{fig:q_kinetic}.
We solve the characteristic equations \eqref{eq:gk_char} using an RK4 method for ODEs.
Note that previous studies have shown that the gyro-averaging can reduce the level of chaos present in a system \citep{del-castillo-negrete_gyroaverage_2012,martinell_gyroaverage_2013,da_fonseca_area-preserving_2014,kryukov_finite_2018}, and hence we focus on the case where the gyro-averaging is applied.
To implement the gyro-average, we take the eikonal approximation \(\mathbf{k}_\perp = n \nabla(\varphi - q \eta)\) and multiply the ballooning mode envelope \(g_n\) by the gyro-average factor \(J_0(k_\perp \rho)\).

To visualize the chaotic transport due to the perturbed electrostatic potential, we plot Poincar\'{e} sections in Figure~\ref{fig:poincare_section}.
Panels (a,b) shows the Poincar\'{e} section for the model gyrokinetic system with the drift wave at the amplitude observed in the simulations for trapped and passing particles corresponding to the \(\mathcal{E}_\perp/\mathcal{E}=2/3\) and \(\mathcal{E}_\perp/\mathcal{E}=1/3\) cases.
Focusing on panel (a), despite presence of chaos throughout the section, the trapped particles mostly remain in two distinct regions of phase space separated by a radial barrier around \(P_{\varphi}/\psi_{lcfs} \approx 0.8\).
The bottom sub-panel shows \(q_{kin}\) plotted against the average \(P_{\varphi}\) for the trajectory, both computed using the weighted Birkhoff average.
Although the rotation numbers are perturbed from the integrable case due to the influence of the drift-wave fluctuations, a non-degenerate minimum still appears in the rotation number plotted against the radial coordinate, showing that this barrier corresponds to a shearless torus.
In contrast, panel (b) shows no clear barrier to radial transport for passing particles.

The effect of the shearless torus can be made more pronounced by reducing the amplitude of the drift wave, which is done in panel (c).
An extremely sharp delineation between the two populations of trapped particles is now visible, suggesting the survival of the invariant shearless torus.
In contrast, in panel (d) we show the Poincar\'{e} section for the model gyrokinetic system with the drift wave at the full amplitude, but without applying the gyroaverage.
In this case, the level of chaos is higher, and the shearless torus has been destroyed, and the two populations of trapped particles are now intermixed.
Consistent with previous studies, we find that gyroaveraging reduces the level of chaos in the system.

To quantify the partial phase space transport barrier effect as well as the dependence on the particle velocity more systematically, we consider the \textit{transmissivity} \(\eta_t\), defined as follows.
We initialize gyrokinetic test particles on the outboard midplane with an initial gyrocenter radial location of \(\psi_n = 0.78\) and a given \(v_{\parallel}\) and \(v_{\perp}\), uniformly distributed in toroidal angle.
Note that the initial \(P_{\varphi}\) of these particles will vary depending on \(v_{\parallel}\).
We take \(\eta_t(v_{\perp},v_\parallel)\) to be the fraction of particles which experience a change in \(\Delta P_{\varphi} / \psi_{LCFS} = 0.04\) over a time window of 4 ms, approximately twice the time window over which the XGC simulations are run.
This mimics the radial transport of a particle across the shearless transport barrier to \(\psi_n = 0.82\) in physical space, accounting for the finite orbit width of the ions.
The behavior of the transmissivity metric in shearless regimes has been studied in a number of works, including \citet{szezech_transport_2009,viana_transport_2021,osorio-quiroga_larmor_2024,grime_effective_2025}.

\begin{figure}
    \centering
    \includegraphics{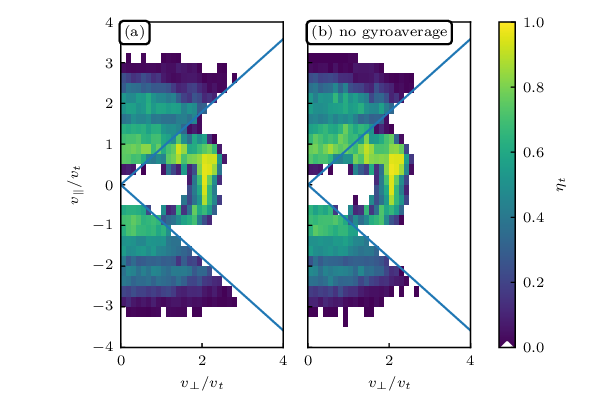}
    \caption{(a-b) Heatmaps showing the transmissivity \(\eta_t\) across the shearless transport barrier region varying \(v_{\parallel}\) and \(v_{\perp}\), with and without gyroaveraging respectively. Unshaded white regions correspond to \(\eta_t = 0\). The trapped/passing boundary in velocity space is shown with blue lines.}
    \label{fig:plot_transmissivity}
\end{figure}

In the presence of unbroken invariant tori, the transmissivity would be exactly zero \(\eta_t = 0\).
Transmissivity \(\eta_t \approx 1\) would indicate fast transport of particles across the shearless region due to chaos over the chosen time window.
Even without unbroken invariant tori, remnants of shearless tori can still exhibit ``stickiness'' with small but non-zero transmissivity.
For example \(\eta_t \lesssim 0.1\) would indicate an effective reduction of transport across the shearless region over the chosen time window of 4 ms.

Plots of the transmissivity \(\eta_t(v_{\perp}, v_\parallel)\) are shown in figure~\ref{fig:plot_transmissivity}, where panel (a) shows the nominal case with gyroaveraging applied to the fields, while panel (b) shows the case with no gyroaveraging applied to the fields.
Focusing first on the nominal case in panel (a), overplotting the trapped/passing boundary \((v_\perp / v_\parallel)^2 = B_{max}/B_{min} - 1\) reveals that the majority of trapped particle phase space experiences a complete suppression of transport across the shearless transport barrier region, strongly suggestive that shearless invariant tori as identified in figure~\ref{fig:poincare_section}(a,c) are generic across most of trapped particle phase space.
Comparing this to the case with no gyroaveraging in panel (b), the picture is extremely similar, although a close inspection in the trapped region of phase space reveals some regions where the transmissivity has become non-zero.

Both cases show a band of much larger transmissivity within the trapped portion of phase space near \(v_\perp/v_t \sim 2\).
Trapped particles in this region of phase space have nearly stationary toroidal precession in the frame rotating with \(\Omega\), indicating strong wave-particle resonance with the model drift wave.
This indicates that although gyroaveraging may change the critical amplitude for the survival of the shearless torus, proximity to resonance may play a more significant role in determining if the shearless torus survives or not.
More generally, one could expect many features of figure~\ref{fig:plot_transmissivity} to originate from the structure of particle phase-space resonances; see for example \citet{antonenas_analytical_2024} for a detailed study of how \(q_{kin}\) determines the structure of the resonances.
The reduction in the size of resonances near shearless regions has been proposed as a mechanism by which shearless tori or their remnants can suppress transport \citep{rypina_robust_2007}.
We defer detailed study of these resonances in response to drift-wave perturbations to future work.

\begin{table}
    \centering
    \begin{tabular}{lccc}
        ~ & passing & trapped & all \\
        \(\overline{\eta_t}\) & 0.5 & 0.27 & 0.33 \\
        \(\overline{\eta_t}\), no gyroaverage & 0.52 & 0.29 & 0.35 \\
    \end{tabular}
    \caption{Thermally-averaged transmissivity \(\eta_t\), over different portions of velocity space.}
    \label{tab:transmissivity}
\end{table}

To determine the net effect of this reduced phase space transport over all particles, we can consider the average \(\overline{\eta_t}\) over a thermal Maxwellian given by
\begin{equation*}
    \overline{\eta_t} = \frac{
        \int_{-\infty}^{\infty} \int_{0}^{\infty} \eta_t(v_\perp, v_\parallel) \operatorname{exp}(-(v_\parallel^2 + v_\perp^2)/(2 v_t^2)) v_\perp \mathrm{d}v_\perp \mathrm{d}v_\parallel
    }{
        \int_{-\infty}^{\infty} \int_{0}^{\infty} \operatorname{exp}(-(v_\parallel^2 + v_\perp^2)/(2 v_t^2)) v_\perp \mathrm{d}v_\perp \mathrm{d}v_\parallel
    }
\end{equation*}
We also consider averages of \(\eta_t\) over a thermal Maxwellian where the integrals are restricted to trapped and passing portions of velocity space respectively.
These averages are reported in Table~\ref{tab:transmissivity} for both the nominal case and the case with no gyroaveraging.
Despite the transmission band for trapped particles around \(v_\perp/v_t \sim 2\), the overall effect of the shearless tori is to reduce the transmissivity by a factor of \(2\) in the trapped portions of phase space compared to the passing portions of phase space.
Trapped particles make up about \(75\%\) of thermal particles at the outboard midplane, suggesting that the presence of shearless tori can lead to a significant drop in the overall radial particle transport due to the presence of shearless tori.
The transmissivity increases slightly without gyroaveraging, although the qualitative picture stays the same.

To connect these results with the observed profiles in the simulations, we posit that the reduction in transmissivity due to shearless tori would likely lead to a drop in the turbulent particle diffusivity \(D_n\).
Meanwhile, the transmission band around \(v_\perp/v_t \sim 2\) suggests that trapped particles with higher energies cross the shearless region more easily, leading to a more pronounced effect in the density transport channel compared to the thermal energy transport channel.
For steady-state density profiles with no core fueling, relevant to this scenario, the presence of peaked density profiles requires the existence of an inward flux of particles, also known as a particle pinch \citep{angioni_particle_2009}.
The density gradient would be set by the balance between the outward diffusive flux \(\sim D_n \nabla n\) and the inward pinch \(\Gamma_{pinch}\), \(\nabla n \sim \Gamma_{pinch} / D\).
Although we have not ruled out a localized increase in the inward particle pinch, the presence of shearless tori suggests a plausible mechanism by which a long-lived increase in ion gyrocenter density gradients can be sustained in the shearless region.

\section{Signatures of Shearless Tori in Gyrokinetic Simulations}\label{sec:shearless_application}

Having established the role of shearless tori in suppressing radial transport in the model test particle dynamics, we now turn to compare the result of the model with the fully self-consistent gyrokinetic dynamics simulated by XGC.
In this section, we begin in \S\ref{subsec:evidence} by presenting direct evidence of signatures of shearless tori in the gyrokinetic simulations visible in the ion gyrocenter distribution function.
In \S\ref{subsec:eddy_detachment} we then study the time-dependent behavior of these phase space signatures and demonstrate a link between the arrest of turbulence propagation at the shearless region with eddy detachment events involving the shearless tori in gyrokinetic phase space.
These signatures provide evidence of the active role that shearless transport barriers play in regulating the turbulence dynamics, and suggest phenomenological approaches for modeling the interaction between turbulence and shearless regions.

\subsection{Direct Evidence of Shearless Tori in Gyrokinetic Simulations} \label{subsec:evidence}

\begin{figure}
	\centering
	\includegraphics{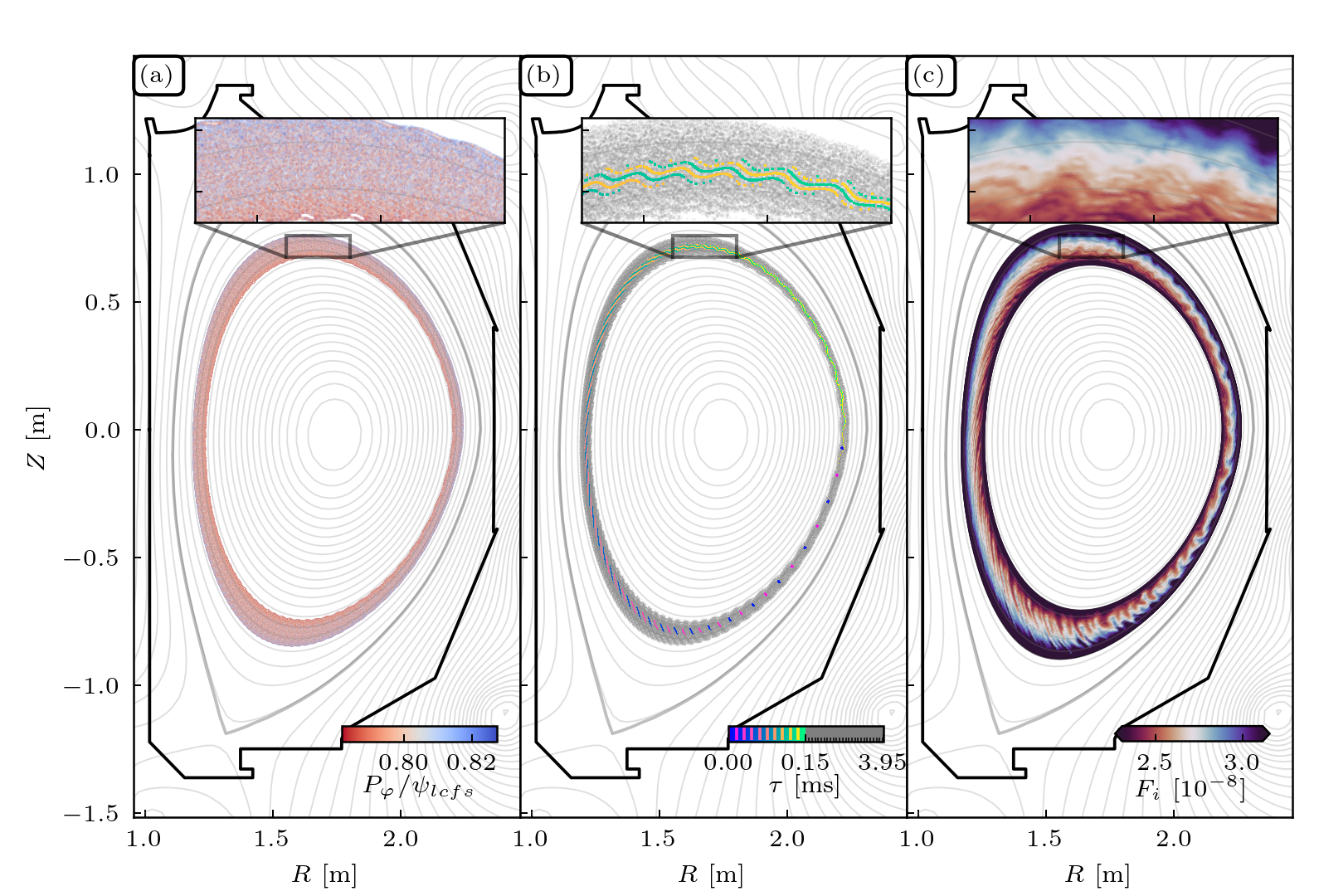}
	\caption{(a) Poloidal surface of section for passing particles in the \(\mathcal{E}_\perp/\mathcal{E}=1/3\) case, initialized uniformly through the domain. The particles are colored by the average radial location of the particles measured by \(P_{\varphi}\), which visualizes radial mixing. (b) Poloidal surface of section initialized with a small blob of particles near the outboard midplane. The particles are colored by their time of flight \(\tau\), which visualizes particle dispersion. (c) Poloidal cross-section of the ion gyrocenter distribution function \(F_i\) from the simulations, restricted to the manifold \(E_{\mu,K}\) corresponding to the Poincar\'{e} sections in (a,b).}
	\label{fig:poloidal_barrier_passing}
\end{figure}

\begin{figure}
	\centering
	\includegraphics{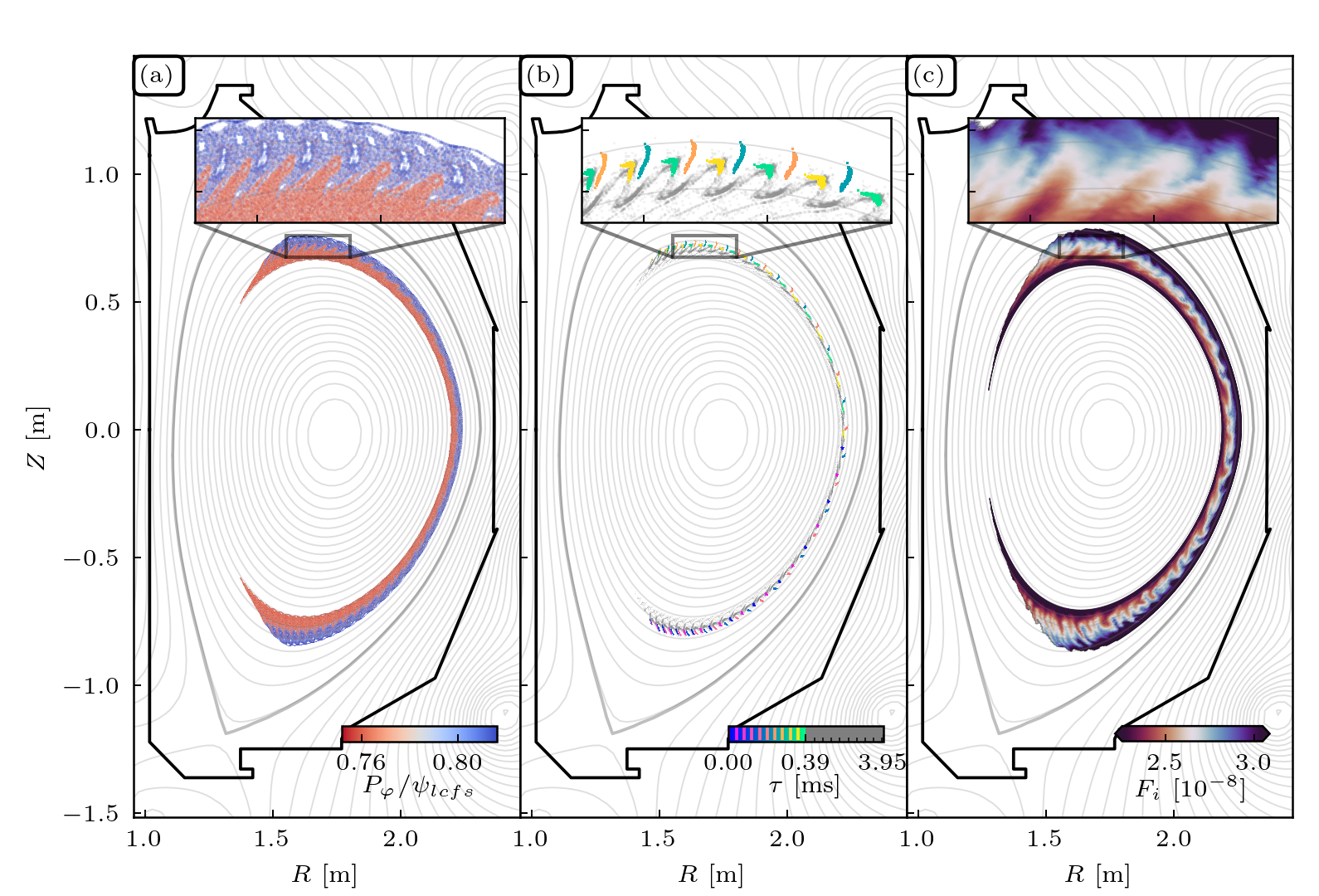}
	\caption{Labeling is same as in figure~\ref{fig:poloidal_barrier_passing}, except for trapped particles in the \(\mathcal{E}_\perp/\mathcal{E}=2/3\) case.}
	\label{fig:poloidal_barrier_trapped}
\end{figure}

Having constructed a model dynamical system and shown the existence of shearless tori, we now compare the model with the gyrokinetic simulation data.
This can be accomplished by considering the ion gyrocenter distribution function \(F_i\), which evolves by the Vlasov equation \eqref{eq:XGC_vlasov}.
In the absence of collisions and sources \(C_i = S_i = N_i = 0\), \(F_i\) will be a scalar Lagrangian invariant of the flow, i.e., it will be constant along the characteristics of the system.
In the trapped region of phase space, the model dynamical system predicts that mixing of characteristics should be suppressed across shearless tori, whereas mixing along tori is unimpeded.
If this effect persists in the gyrokinetic simulation, we should expect to see iso-contours of \(F_i\) trace out meandering paths similar to the shape of the shearless torus.
Meanwhile, in the passing region of phase space, we should expect iso-contours of \(F_i\) to cross the shearless region, as no shearless torus is predicted to exist in that region.
To visualize this comparison, we can examine the distribution of \(F_i\) projected onto \(E_{\mu,K}\) and look for evidence of the effect of shearless tori.
In XGC, \(F_i\) is computed as $F_a+F_g+F_p$, where $F_a$ is an analytic function (Maxwellian for ions and Maxwellian-Boltzmann for electrons), $F_g$ is stored on the grid, and $F_p$ is carried by the particles \citep{ku_new_2016}.

To perform this comparison, we focus on poloidal sections, i.e. the intersection of \(E_{\mu,K}\) with a poloidal plane.
We make a technical remark that the poloidal surface of section is not a true Poincar\'{e} section for trapped particles, as they do not always pass through poloidal planes transversally.
While relevant for computing topological information such as the rotation number, this is not relevant for visualization purposes.
In figures~\ref{fig:poloidal_barrier_passing} and \ref{fig:poloidal_barrier_trapped}, we plot a sequence of poloidal surfaces of section in the (a) and (b) panels, as well as a comparison against \(F_i\) in the (c) panels.
In the (a) panels, we initialize the surface of section with particles uniformly throughout the domain and color the particles by their average location, which visualizes radial mixing similar to the Poincar\'{e} sections in figure~\ref{fig:poincare_section}.
Meanwhile, in the (b) panels, we initialize the surface of section with a small blob of particles near the outboard midplane.
Here, we use color to show the time evolution of the blob, visualizing particle dispersion.
Note that in the trapped particle case, both the co-passing and counter-passing branch is shown for the first banana orbit, but subsequently only the counter-passing parts of the orbit are shown.
Finally, the (c) panels show \(F_i\), restricted to the intersection of \(E_{\mu,K}\) with a poloidal plane.

We start by considering passing particles in the \(\mathcal{E}_{\perp}/\mathcal{E} = 1/3\) case in figure~\ref{fig:poloidal_barrier_passing}.
Similar to figure~\ref{fig:poincare_section}(b), panel (a) shows no clear barrier to radial mixing.
As in figure~\ref{fig:poincare_section}(b), we see no barrier to radial mixing in the shearless region.
The reason for this can be understood through panel (b), which shows that nearby particles quickly disperse within a single poloidal transit due to the influence of magnetic shear.
In panel (c) we see that \(F_i\) has a finely filamented structure which mirrors the short-time particle dispersion of the passing particles in the model dynamical system, illustrating the lack of any radial transport barrier in the passing particle regions of phase space.

Now we consider trapped particles in the \(\mathcal{E}_{\perp}/\mathcal{E} = 2/3\) case in figure~\ref{fig:poloidal_barrier_trapped}.
Similar to figure~\ref{fig:poincare_section}(a), panel (a) shows two populations of trapped particles with suppressed radial mixing across the shearless torus.
This behavior is made even more evident in panel (b), where a small blob of trapped particles is initialized near the outboard midplane.
The gray band, which visualizes the dispersion of the particles after many bounce orbits, approximately follows the meandering curve which separates the two populations of trapped particles in panel (a).
This shows that particles easily disperse along the shearless torus, but radial dispersion across the torus is suppressed.
We remark that this particle dispersion behavior does not actually require the presence of an unbroken shearless torus, as it has been observed that the remnants of shearless invariant tori often continue to act as ``sticky'' sets which suppress transversal transport across the former tori \citep{szezech_transport_2009}.
Finally in panel (c), we observe that the level sets of \(F_i\) have a clear meandering structure similar in shape to the barriers observed in panels (a) and (b), reflecting the presence of a shearless phase-space transport barrier.
This provides strong evidence that shearless tori observed in the model dynamical system are able to survive as shearless phase-space transport barriers in the high-fidelity gyrokinetic simulations.

\subsection{Shearless Barrier Breakdown via Eddy Detachment}\label{subsec:eddy_detachment}

Having established the presence of shearless phase-space transport barriers, we now turn to the question of the mechanism which stops the radial propagation of avalanches across the shearless region.
In this section, we show how avalanches impinging on the shearless tori cause eddy detachment, reminiscent of the formation of warm and cold core ring structures in oceanic jets.

\begin{figure}
	\centering
	\begin{tabular}{c}
		\includegraphics[width=0.49\linewidth]{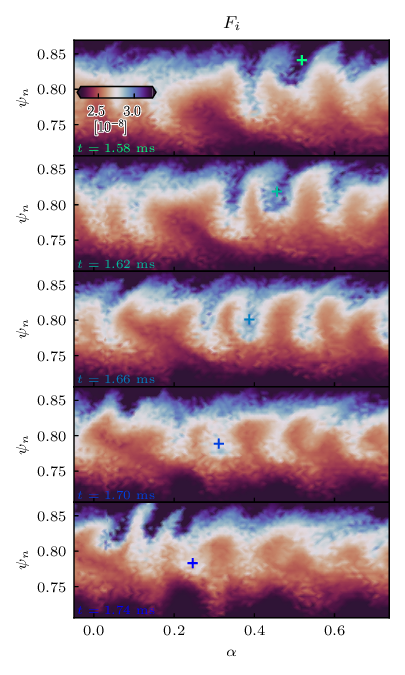}
	\end{tabular}%
	\begin{tabular}{c}
		\includegraphics[width=0.49\linewidth]{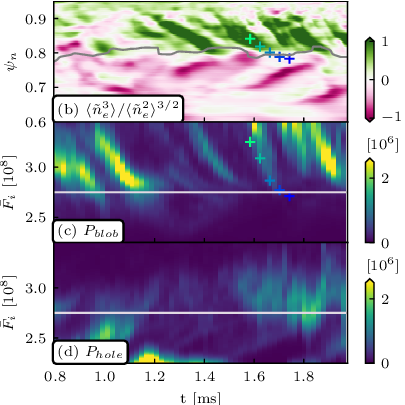} \\
		\includegraphics[width=0.35\linewidth]{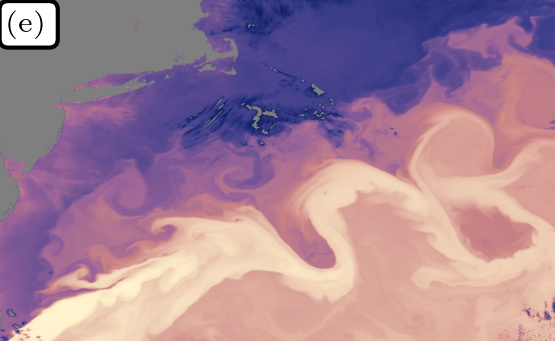}
	\end{tabular}
	\caption{(a) Temporal snapshot sequence of ion GC distribution function \(F_i\). The timestamp is shown in colored text, and a single eddy detachment event is tracked by the `+' symbols.
	(b-d) Sequence of panels showing the density skewness along with the results from the persistent homology analysis of blob and hole structures. The `+' symbols indicate the same eddy detachment event as in the left panel.
	(e) Image of the Gulf Stream sea surface temperature showing eddy detachment, leading to the formation warm and cold core rings (source: https://earthobservatory.nasa.gov/images/5432/the-gulf-stream)}
	\label{fig:blob_phase_space}
\end{figure}

In Figure~\ref{fig:blob_phase_space}(a), we show the time evolution of the ion gyrocenter distribution function \(F_i\) for trapped particles in the \(\mathcal{E}_{\perp}/\mathcal{E}=2/3\) case, restricted to \(E_{\mu,K}\) as before.
For ease of viewing, the data is projected from a poloidal section to flat coordinates \((\alpha, \psi_n)\), where \(\alpha = \varphi - q \theta\) is the perpendicular field line label.
Focusing on the propagation of structures from the edge towards the core, we observe several tongue-like structures that detach into blobs that then quickly dissipate.
When restricted to \(E_{\mu,K}\), \(F_i\) increases radially outward, so phase space structures propagating from the edge correspond to phase space blobs here.

Tracking the evolution of one of these perturbations in time and radius and overplotting onto the density skewness in Figure~\ref{fig:blob_phase_space}(b), we observe that these structures precisely track the radial propagation of the avalanches.
Moreover, notice that the blob detachment event corresponds with the termination of the density skewness avalanche at the shearless region.
We remark that this behavior is highly reminiscent of the eddy detachment observed in oceanic jets such as the Gulf Stream \citep{the_ring_group_gulf_1981,olson_rings_1991}.
An example of eddy detachment in the Gulf Stream is shown in Figure~\ref{fig:blob_phase_space}(e).
The meandering of the Gulf Stream, associated with north-south intrusions of warmer and colder water from the tropics and arctic respectively, pinches off and detaches to form eddy rings with cores of warm or cold water detectable by satellite imaging.

\begin{figure}
	\centering
	\includegraphics[width=\linewidth]{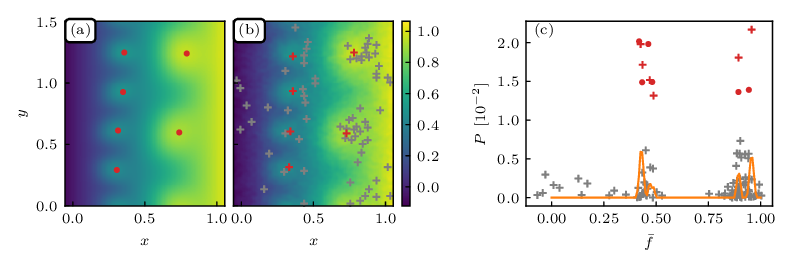}
	\caption{Example 2d function with blobs, both (a) without and (b) with noise applied. Local maxima are indicated with points in (a), and with crosses `+' in (b). (c) Persistence diagrams computed by sublevel set persistent homology for the original (red points) and noisy functions (red and gray crosses). The silhouette for the noisy case is overplotted (orange curve), indicating two bands of local maxima.}
	\label{fig:blob_persistence}
\end{figure}

In order to demonstrate that this phase space blob detachment phenomenon is a robust feature of the data, we proceed to study the statistical properties of local maxima and minima in the ion gyrocenter distribution function \(F_i\).
To illustrate the challenge of detecting `blobs' in noisy datasets, we turn to figure~\ref{fig:blob_persistence}.
Panel (a) shows a function \(f\) with two bands of blobs at different values of \(x\). 
In the absence of noise, these blobs can be cleanly identified by finding the local maxima of \(f\), marked with red points.
However, panel (b) shows the same function \(f\) with high-frequency noise added, which may represent noise from particle discreteness or fluctuations from smaller-scale physical processes.
Although the two bands of blobs appear to persist, many new local maxima marked with `+' symbols appear due to the presence of noise.

To deal with this challenge, we apply tools from topological data analysis (TDA), namely sublevel set persistent homology \citep{mirth_representations_2021}.
Sublevel set persistent homology provides a quantitative way to address the question of which local maxima and minima of a function are ``robust'', as well as providing algorithms for efficiently identifying these features.
The method is based on computing the homology of sublevel sets, i.e. the number of connected components and holes present in the set of points \((x,y)\) such that \(f(x,y) \le z\) for a given \(z\).

To explain how this method quantifies the robustness of local maxima and minima, we give a brief intuitive explanation.
If one imagines that \(f(x,y)\) measures the height of a landscape and \(z\) to be the level of water filling the landscape, then local maxima of \(f\) with height greater than \(z\) correspond to peaks that form landmasses breaking the surface of the water.
For each local maximum, there will exist a range of water heights \(f_{birth} \le z < f_{death}\) where the local maximum will be the tallest peak on its own landmass, where \(f_{birth}\) and \(f_{death}\) are the values of \(f\) for which the landmass is distinct.
These values define the \textit{persistence} \(P = (f_{death}-f_{birth})/2\), which measures (half) the range of \(z\) values for which the local maximum is the tallest peak on its own distinct landmass.
While new local maxima can easily be created by noise, exactly one local maximum will inherit the title of being the tallest peak on a given landmass at a time.
Thus, counting high-persistence features gives a way to quantify the number of significant local maxima.
A similar intuitive explanation can be constructed for local minima by considering the deepest trench in each distinct body of water.

To use this quantity, the persistence can be plotted against the average birth and death value \(\bar{f} = (f_{birth}+f_{death})/2\) of each feature to create a \textit{persistence diagram}.
In the case where \(f\) has an overall gradient in \(x\), \(\bar{f}\) is correlated with the \(x\) location of the feature.
Note that persistence diagrams also sometimes plot \(f_{birth}\) against \(f_{death}\), which is simply a 45-degree rotated version of the type of diagram used here.
The persistence diagrams for both the noisy and non-noisy example \(f\) are shown in figure~\ref{fig:blob_persistence}(c).
The two bands of blobs are visible as a cluster of high-persistence features detected around \(\bar{f} \approx 0.4\) and \(\bar{f} \approx 0.8\).
Although many more local maxima are present in the noisy function, both persistence diagrams show the same number of high-persistence local maxima corresponding to the two bands of blobs, marked in red.
This persistence behavior can be rigorously understood through the stability of persistence diagrams under perturbation \citep{cohen-steiner_stability_2007}.
Physically, \(P\) measures (half) the fluctuation amplitude \(\delta f\) which is organized into a topologically distinct `blob' of \(f(x,y)\).

In order to give a statistical summary of the topological information provided by the persistence diagram into a single curve, we use \textit{silhouettes} \citep{chazal_stochastic_2014}.
These are constructed by taking a weighted average of tent functions associated with each feature in the persistence diagram.
Intuitively, the silhouette is a function which peaks where the persistence diagram is dominated by high-persistence features, and is low where the persistence diagram is dominated by low-persistence features.
We use a power weighting of \(P^2\) to emphasize the high-persistence features.
With this weighting, the silhouette \(\propto |\delta f|^3\), where \(\delta f\) is again the fluctuation amplitude associated with topologically distinct blobs or holes.
An example silhouette is shown in figure~\ref{fig:blob_persistence}(c), showing how the two bands of blobs are captured by peaks in the silhouette curve at the values of \(\bar{f}\) corresponding to the radial location of the blobs.

Using these silhouettes, we can track how topological features of the ion gyrocenter distribution function \(F_i\) correlate with the avalanches observed in the physical fields.
We focus on the detection of local maxima and minima of the ion gyrocenter distribution function \(F_i\) in each of the 16 computational toroidal planes over a range of values of time \(t\).
When restricted to \(E_{\mu,K}\), this corresponds to finding local maxima and minima of planar functions.
We use the simplex tree and representations modules of the GUDHI library \citep{boissonnat_simplex_2014,gudhi:FilteredComplexes,gudhi:PersistenceRepresentations} to compute the persistence diagrams and silhouettes.
The results of this analysis is shown in panels (c) and (d) of Figure~\ref{fig:blob_phase_space}, where the silhouettes are plotted in a color panel.
\(P_{blob}\) and \(P_{hole}\) correspond to the persistence of local maxima and minima of \(F_i\), and stronger persistence is shown with brighter colors.
Following the earlier discussion, peaks in the persistence along the vertical axis of these plots correlate with radial location of bands of maxima/minima in \(F_i\), while evolution of the peaks in the horizontal axis give the time evolution of these radial bands.

Focusing on panel (c), the phase space blobs show clear correlation with incoming avalanches from the edge, showing that these avalanches are tied to local topological features of the ion gyrocenter distribution function.
These blobs consistently terminate at a particular value of \(F_i \approx 2.75\), which corresponds to the value of \(F_i\) aligned with the shearless torus.
This indicates that the shearless phase-space barrier eddy detachment event observed in panel (a) is a generic feature of inward-propagating avalanches.
Note that panel (d) shows faint bands of outward propagating phase space holes, although these are not as pronounced as the inward propagating phase space blobs.
This is consistent with the expectation that turbulence spreading would be more relevant to the larger amplitude fluctuations coming from the edge.

We conclude this section with a brief speculation about the potential mechanism by which shearless transport barriers could lead to the arrest of turbulence propagation.
In \citet{pratt_dynamics_1986,pratt_meandering_1988} it was proposed that eddy detachment in the Gulf Stream could be understood through the idealized model following the contour dynamics of a single potential vorticity (PV) front.
The key process that leads to eddy detachment in the model is vortex induction due to the north-south displacement of the PV front.
The strong PV jump across the jet correspondingly leads to a strong nonlinear vortex flow induced by displacements of the jet, leading to wave breaking and subsequent eddy detachment.
We remark that this can also be interpreted as a type of reconnection for meandering tori, potentially related to the breakup of shearless tori in dynamical systems models \citep{wurm_meanders_2005}.

In magnetized plasmas, the PV can be identified with the ion gyrocenter density \citep{mcdevitt_poloidal_2010,gurcan_zonal_2015}.
Thus, the situation observed in the XGC simulations, namely a sudden increase in the ion gyrocenter density gradient across the zonal \(E \times B\) jet, is closely analogous to the situation of the PV jump across the oceanic jet.
Physically, displacements of the plasma across the \(E \times B\) jet would correspond to a sudden intrusion of lower or higher ion gyrocenter density across the shearless region.
Since the adiabatic parallel electron response is relevant for the outer core, the expectation would be that this sudden density change would induce significant blob spin \citep{myra_convective_2004}, i.e. vortical \(E \times B\) `eddy' flow about the blob (or hole) center.
Spin mitigates curvature-induced charge polarization in blobs, which is precisely the process responsible for driving the ITG instability, suggesting a potential physical mechanism for the arrest of turbulence propagation at the shearless region.
For future work, it may be interesting to use blob seeding experiments in simulations \citep{cheng_transport_2023} to perform controlled numerical experiments to more precisely determine the ways in which shearless regions impact the propagation of turbulence.

\section{Discussion and Outlook} \label{sec:discussion}

To summarize up to this point, this work demonstrates, for the first time in a high-fidelity simulation with realistic geometry and profiles, the presence of a shearless transport barrier in gyrokinetic drift-wave turbulence.
The genericity of the conditions under which shearless tori appear, near non-degenerate minima and maxima of the zonal \(E \times B\) rotation rate, suggests that the structures observed here may be generic to a broad class of plasma turbulence regimes.
Establishing the exact parameter regimes in which shearless transport barriers could exist in realistic scenarios is an open question.

In this section, we begin in \S\ref{subsec:assumptions} by discussing the fundamental assumptions of the test particle map model, namely the usage of collisionless dynamics and the single-mode model perturbation, and how they might be relaxed.
Then, in \S\ref{subsec:extensions} we conjecture other regimes of turbulence where shearless transport barriers may play a role and discuss their potential impact on questions of relevance for fusion energy, suggesting pathways for future work.

\subsection{Applicability to Collisional and Broadband Turbulence} \label{subsec:assumptions}

We first remark on the applicability of the model dynamical system to the gyrokinetic simulation data, as the model involves a rather severe truncation of the electrostatic fluctuations present in the gyrokinetic simulation.
The timescale associated with one iteration of the test particle map model for the trapped particles is equal to the ion bounce time \(\omega_b \approx 15\) kHz.
The ion collision rate \(\nu_{i} \approx 1.2\) kHz is a modest fraction of \(\omega_b\), suggesting weakly collisional dynamics is applicable.
Meanwhile, from figure~\ref{fig:hovmoller} it can be seen that the drift wave envelope evolves on the timescale \(\sim 20\) kHz, which is also comparable to \(\omega_b\).
Despite the severe truncation of the electrostatic fluctuations and weak timescale separation between the turbulence envelope evolution timescale and the test particle map timescale, the model dynamical system is still able to capture essential qualitative features of the gyrokinetic simulation data.

We first discuss how the reliance of shearless transport barriers on collisionless test particle dynamics used in the model can be relaxed.
Shearless transport barriers were originally studied in the fluid limit of drift wave turbulence, whereas the collisionless limit studied here focused on trapped particle dynamics which are completely absent from the fluid limit.
Thus, this suggests that shearless transport barriers could play different roles throughout a range of collisionalities in plasma turbulence.
To identify shearless transport barriers in regimes with higher collisionality, we highlight two theoretical approaches.
The first is the relationship between the local Cauchy-Green tensor, which integrates the total local deformation of a Lagrangian parcel, and diffusion.
In fluid models, a certain minimum deformation principle can be used to define a variational notion of shearless material transport barrier \citep{farazmand_shearless_2014}, across which the combined effect of diffusion and stretching due to advection will be minimum.
Another approach is to use stochastic Lagrangian trajectories, in which the effect of both collisions and advective transport is accounted for through a stochastic perturbation to the test particle dynamics.
In studies including \citet{kwon_global_2000,cao_maintenance_2024}, it is noted that shearless tori can continue to act as effective transport barriers in the presence of collisions.

Now, we discuss the reliance of the test particle map on the model electrostatic single-mode perturbation.
As discussed earlier, the assumption of electrostatic transport is superfluous, as the existence of shearless tori depends only on the presence of a non-degenerate minimum or maximum in the kinetic safety factor \(q_{kin}\).
This is a topological quantity which for passing particles closely follows the magnetic safety factor \(q\), and for trapped particles more closely follows the \(E \times B\) rotation rate \(\Omega_E\).
The key feature of the single-mode approximation was the dimensionality reduction provided by the extra gyrokinetic invariant \(K = H - \Omega P_{\varphi}\).
This would correspond physically to a regime of weak dispersion, where particles approximately ``see'' perturbations that rotate rigidly with a toroidal rotation rate of \(\Omega\).
We will discuss first the validity of the single-mode approximation, and then discuss possible ways to relax the need for weak dispersion.

The first remark is on the relationship between the single-mode assumption and particle stochasticity.
Particle stochasticity is a major pillar of quasilinear theory \citep{diamond_modern_2010}, and typically a broad spectrum of modes is invoked in order to satisfy the resonance overlap criterion necessary for particle stochasticity.
For the toroidal mode number \(n=39\) of the ballooning mode studied here, rational surfaces are closely spaced.
The representation \eqref{eq:ballooning_mode} keeps a broad range of poloidal harmonics active, and evidently from the Poincar\'{e} sections in figure~\ref{fig:poincare_section} the island overlap criterion is satisfied over most of the region of interest.
Thus, even a single tertiary mode is able to drive significant transport by itself, potentially relaxing driving gradients and temporarily ``outcompeting'' other tertiary modes in the immediate vicinity.
Since wave-kinetic models would tend to predict the localization of drift-wave modes around minima and maxima of the \(E \times B\) rotation rate, this could provide a generic mechanism by which energy is accumulated into a small number of tertiary modes, enabling the usage of a single- or few-mode model.
For example, \citet{cao_maintenance_2024} demonstrates how chaotic mixing could reinforce the amplitude of vortical flows associated with tertiary modes in a fluid slab model of resistive drift-wave turbulence.

The next remark is on the relationship between the single-mode assumption and `fully developed turbulence', in the sense of a broad spectrum of dispersive modes being active and nonlinearly exchanging energy.
Geophysical observations such as in figure~\ref{fig:blob_phase_space}(e) show that large-scale meanders of jets are fully compatible with much smaller fine-scale structures being present.
While it is unlikely single-mode test particle dynamics would be able to predict quantitative levels of transport due to the combination of both large-scale meanders and small-scale turbulent dispersion, qualitative features such as the presence of a north-south barrier to mixing due to the jet can be reproduced using models that feature severe truncations of the active mode spectrum.
For example, even in the presence of fully-developed spectra of turbulence featuring energy and enstrophy cascades, \citet{cao_rossby_2023} shows how dynamical systems models can still capture key features of turbulent mixing and small-scale statistics in Rossby wave turbulence, analogous to drift-wave turbulence in plasmas.

Finally, we turn to possible relaxations of the weak dispersion hypothesis, which stemmed from the requirement for the perturbations to be approximated by a rigid toroidal rotation \(\Omega\).
There exist definitions of shearless transport barriers in planar flows for finite-time systems with arbitrary time dependence \citep{beron-vera_invariant-tori-like_2010,farazmand_shearless_2014,falessi_lagrangian_2015}, which can act as effective barriers even for waves with chaotically forced amplitudes.
However, arbitrary time dependence of the perturbations will also break the conservation of the invariant \(K\), meaning the reduction to planar dynamics will no longer be exact.

As discussed in \S\ref{subsec:planar_map}, \citet{gonzalez-enriquez_singularity_2014} develops a systematic classification for shearless invariant tori in higher dimension using singularity theory.
In their classification, non-degenerate shearless tori in planar maps correspond to the ``fold'' singularity, originating from the fact that the lowest-order Birkhoff normal form of the Hamiltonian is cubic in the action variable in the vicinity of a non-degenerate shearless torus.
In higher dimension, shearless tori occur when the local shear \(D\omega(J)\) is singular, or equivalently when \(\operatorname{det}(D\omega(J)) = 0\).
Intuitively, one can expect that the zero level set of the scalar function \(\operatorname{det}(D\omega(J))\) over the action variables \(J\) as well as any parameters \(\lambda\) in the system forms a co-dimension 1 manifold, which is indeed the case for the fold singularity.

Since invariant tori are usually \(N\)-dimensional for \(2N\)-dimensional Hamiltonian systems, there is normally no expectation that surviving KAM tori for the usual twist case will act as global barriers to transport.
Intuitively, it is easy for paths in \(2N\)-dimensional phase space to avoid \(N\)-dimensional barriers when \(N > 1\).
The importance of the fold singularity is that it implies that shearless tori in higher dimension are organized into co-dimension 1 families.
Conjecturing that higher-dimensional (\(N > 1\)) invariant tori in co-dimension 1 fold families are ``robust'' to perturbations like their planar map (\(N=1\)) cousins, then surviving shearless KAM tori could act as global barriers to transport.
In collisionless gyrokinetics, where the gyrophase is an ignorable coordinate and the magnetic moment \(\mu\) can be taken as a fixed parameter of the particle dynamics, the \(N=2\) case would be relevant.

We also conjecture that the persistence of fold tori may be related to the ability of the model electrostatic perturbations to capture qualitative features of \(F_i\) in figure~\ref{fig:poloidal_barrier_trapped}.
Although the single mode approximation represents a severe truncation of the fluctuation spectrum observed in the simulations, the addition of other modes with sufficiently small amplitude may simply perturb the family of shearless phase-space transport barriers observed in the test particle model over different values of \(v_\parallel,v_\perp\) without destroying them.
Given advances in computational methods for detecting invariant tori in higher-dimensional systems, see e.g. \citet{haro_parameterization_2016,gonzalez_efficient_2022,ruth_robust_2025}, it may be within reach to explore the robustness of such structures, which may be a fertile avenue for future work.

\subsection{Extension to Other Turbulence Regimes} \label{subsec:extensions}

Although this work has focused primarily on ITG turbulence within the outer core, we remark that there do not appear to be serious roadblocks to developing extensions of the work here to look for shearless transport barriers in other turbulence regimes.
We discuss some of these regimes and how including the effect of shearless transport barriers could improve modeling for fusion reactors.

Shearless transport barriers may be particularly relevant in trapped electron mode (TEM) dominated turbulence, as the passing electrons remain primarily adiabatic while the radial transport of trapped electrons plays the main role.
For core turbulence, several studies have observed the formation of quasi-coherent trapped electron modes (QC-TEMs) \citep{arnichand_quasi-coherent_2014,hornung_eb_2017,lee_observation_2018} in the low-collisionality L-mode confinement regimes, which also correlate with stronger `staircase'-like behavior observed both in experiments and simulations.
It may be interesting to produce test-particle map models of transport induced by QC-TEMs and see if shearless transport barriers form in these regimes.
For temperature-gradient driven modes, it appears that shearless transport barriers play a larger role in suppressing particle transport.
Thus, including the effect of shearless regions on core turbulence dynamics could thus potentially improve the modeling of density profiles, especially density peaking, in low-collisionality regimes.

Moving on to pedestal turbulence, a wide array of ion-scale modes can be destabilized including microtearing modes (MTMs), ITGs, trapped electron modes (TEMs), kinetic ballooning modes (KBMs), and drift-Alfv\'{e}n waves (DAWs) \citep{kotschenreuther_gyrokinetic_2019,diallo_review_2020}.
Moreover, there are several scenarios where pedestal-localized modes are observed to organize into weakly or quasi-coherent bands, the latter sometimes being referred to as ``washboard modes'' if multiple bands are visible.
Fluctuations are also often observed to radially localize somewhere close to strong gradient region, which again corresponds to the shearless region.
MTMs in H-mode \citep{hatch_microtearing_2016,hatch_microtearing_2021} have been observed to occur as washboard modes, with their radial localization determined by a well in the electron diamagnetic frequency.
The weakly coherent mode (WCM) in I-mode \citep{liu_physics_2016,herschel_experimental_2024}, although its nature and role in regulating transport in I-mode is still unclear, is also a weakly coherent electromagnetic mode which is localized to the pedestal in the presence of a temperature barrier.
In quiescent H-modes (QH-mode) and wide pedestal quiescent H-modes (WPQH-mode), fluctuation measurements also detect fluctuations localized to the \(E_r\) well and linear analyses have shown the existence of modes which match this degree of radial localization \citep{li_numerical_2022}.

Transport due to any of the aforementioned modes may be affected by the presence of shearless transport barriers.
Addressing the question of how turbulent transport can be suppressed in shearless regions may inform the control of turbulence address ELM mitigation and heat exhaust.
For example in QH-mode and WPQH-mode, turbulence in the shearless region has been implicated in broadening the pedestal and the scrape-off-layer heat flux width \citep{chang_role_2024,ernst_broadening_2024,li_role_2024}.
Including the effect of shearless regions on pedestal turbulence may help explain how steep gradient regions are maintained in the presence of microturbulence, and perhaps also shed additional light in the dynamics of turbulence spreading to/from the scrape-off-layer.
\citet{osorio-quiroga_shaping_2023} observed that as the depth of the \(E_r\) well or hill increased, the critical amplitude for the breakup of the shearless torus increased.
However, their study did not include how wave trapping might affect the perturbations, either through changes in the mode frequency or radial structure.
We remark that a combination of wave trapping physics with barrier breakup studies may be a fruitful avenue for future work to understand the role of zonal flow curvature, as the curvature affects both the linear physics of eigenmode structure as well as the nonlinear physics of chaotic transport.

Finally we remark that in low collisionality regimes in the pedestal, the bootstrap current may lead to a weakening of the magnetic shear due to the edge pressure gradients.
This may suggest a role for shearless transport barriers in the passing regions of phase space, rather than just in the trapped regions of phase space.
Improvements in confinement have also been observed in core turbulence regimes with weak or reversed magnetic shear.
It would be interesting to see if shearless transport barriers might play a role in suppressing turbulent transport in the shearless region for these specific scenarios.

\section{Summary} \label{sec:summary}

In summary, we have shown that shearless regions associated with zonal \(E \times B\) jets can act as shearless transport barriers, leading to suppression of transport contrary to expectations from local shear suppression of turbulence.
These shearless regions are non-degenerate, in the sense of having non-zero zonal flow curvature.
This was supported by a concrete example of a shearless transport barrier identified in the outer core region of a global gyrokinetic simulation from XGC with realistic tokamak geometry and profiles.
This barrier demonstrated a significant localized increase in the ion gyrocenter density gradient, and also acted as a barrier to turbulence propagation between the inner core and the edge.

We demonstrated how a map model of the collisionless gyrokinetic test particle dynamics was able to identify shearless tori that corresponded with this shearless transport barrier.
This model was constructed by extracting the dominant fluctuation mode in the shearless region, which was an intermediate-\(n\) (\(n=39\)) electrostatic ballooning mode with a finite radial envelope localized to the zonal jet.
Non-monotonic behavior in the kinetic safety factor profile, \(q_{kin}\), was identified as a necessary ingredient for shearless transport barriers to exist.
In particular, under the usual conditions of magnetic shear, this condition is only satisfied for the trapped particle portion of phase space.
Poincar\'{e} sections were used to visually identify the presence of shearless invariant tori in the model, and direct signatures of shearless tori were identified in the ion gyrocenter distribution function \(F_i\) taken from XGC.
A transmissivity analysis showed that the net effect of the shearless tori was to reduce radial particle transport by a factor of \(\sim 2\) in the trapped portions of phase space, potentially explaining the observed increase in ion gyrocenter density gradients.
Finally, the arrest of turbulence propagation was linked to phase-space barrier eddy detachment events in \(F_i\), and ``blob spin'' was proposed as the physical mechanism for this arrest.

The key conclusion of this work is that shearless transport barriers, rather than being relegated to simple map models, can play an active role in the dynamics of high-fidelity gyrokinetic simulations.
Although turbulence is unlikely to exactly satisfy the conditions necessary for shearless invariant tori to exist in the usual dynamical systems sense, map models of the gyrokinetic test particle dynamics can be constructed by extracting out the dominant fluctuations from the turbulence.
Even despite severe truncations of the electrostatic fluctuations, these models are still able to match qualitative features of the simulations and identify the presence of these shearless transport barriers.
Such models may inform how shearless regions, which are a nearly generic feature of sheared \(E \times B\) zonal flows, impact particle transport and turbulence spreading in both the core and edge.

\section*{Acknowledgements}
The authors thank D.R. Hatch and P.J. Morrison for fruitful comments on this work.
Several anonymous individuals are also thanked for their feedback.
This work was supported by the US Department of Energy, Office of Science (N.M.C. contract number DE-FG02-04ER54742), (H.Z. and T.S. contract number DE-AC02-09CH11466); and the S\~{a}o Paulo Research Foundation (FAPESP), Brasil (G.C.G. Process Number \#2024/02591-0).
This research used resources of the National Energy Research Scientific Computing Center (NERSC), a Department of Energy User Facility using NERSC award FES-ERCAP002927.

\section*{Declaration of Interests}
The authors report no conflict of interest.

\section*{Disclaimer}
This report was prepared as an account of work sponsored by an agency of the United States Government. Neither the United States Government nor any agency thereof, nor any of their employees, makes any warranty, express or implied, or assumes any legal liability or responsibility for the accuracy, completeness, or usefulness of any information, apparatus, product, or process disclosed, or represents that its use would not infringe privately owned rights. Reference herein to any specific commercial product, process, or service by trade name, trademark, manufacturer, or otherwise does not necessarily constitute or imply its endorsement, recommendation, or favoring by the United States Government or any agency thereof. The views and opinions of authors expressed herein do not necessarily state or reflect those of the United States Government or any agency thereof.

\section*{Data Availability Statement}
The code that support the findings of this study are openly available in https://github.com/Maplenormandy/c1lgkt

\appendix

\section{Nontwist Tori in Continuous Flows}\label{app:nontwist_tori}

In this appendix, we prove a brief result showing that whether or not an invariant torus in a 3-dimensional continuous flow is nontwist is invariant to the choice of 2-dimensional Poincar\'{e} section used to determine if it is nontwist.

Let \(E\) be a 3-dimensional manifold with a vector field \(X\) defining a continuous flow on \(E\).
Suppose there exists a (local) smooth foliation of \(E\) into invariant 2-tori \(K_p\) indexed by a single action-like variable \(p\), with the tori parameterized by angle variables \(\vartheta = (\vartheta_1, \vartheta_2)\).
Denote the associated torus frequencies as \(\omega(p) = (\omega_1, \omega_2)\).
For the gyrokinetic particle dynamics studied here, the torus frequencies \(\omega_1,\omega_2\) can be taken as the toroidal and poloidal transit frequencies \(\Omega_{\varphi'},\Omega_{\theta}\), and \(p\) can be taken as the toroidal canonical angular momentum \(P_{\varphi}\).

Suppose \(\Gamma \subset E\) is a Poincar\'{e} section such that the intersection \(\gamma_p = K_p \cap \Gamma\) is a connected curve for each \(p\).
\(\gamma_p\) will be a 1-torus; let \(n = (n_1, n_2)\) be the winding numbers of \(\gamma_p\) around \(K_p\).
The outboard midplane section corresponds to \(n=(1,0)\).
The rotation number \(\iota(p)\) associated with \(\gamma_p\) under the action of the Poincar\'{e} map \(f: \Gamma \to \Gamma\) can be computed
\begin{equation} \label{eq:iota}
	\iota(p) = \frac{1}{n^2} \frac{n \cdot \omega}{n^\perp \cdot \omega}
\end{equation}
where \(n^\perp = (-n_2, n_1)\).
In the following proposition, we show that if a torus is nontwist in any particular choice of section \(\Gamma\), it will be nontwist in all choices of section \(\tilde{\Gamma}\):

\begin{proposition}
	Suppose \(\gamma_0\) is an non-degenerate irrational nontwist torus of the Poincar\'{e} map \(f: \Gamma \to \Gamma\) associated with \(\Gamma\), meaning \(\iota'(0) = 0\), \(\iota''(0) \neq 0\), and \(\iota(0)\) is irrational. Then, if \(\Gamma'\) is another Poincar\'{e} section such that \(\tilde{\gamma}_0 = K_0 \cap \tilde{\Gamma}\) is simply connected, then \(\tilde{\gamma}_0\) will be a non-degenerate irrational nontwist torus of \(f': \tilde{\Gamma} \to \tilde{\Gamma}\)
\end{proposition}

\begin{proof}
	First, note from \eqref{eq:iota} that \(\iota(0)\) being irrational occurs if and only if the frequency vector \(\omega(0)\) is irrational.
	Then, for any other winding numbers \(\tilde{n}\), the rotation number \(\tilde{\iota}(0)\) be irrational if and only if \(\iota(0)\) is irrational.
	
	To show \(\tilde{\gamma}_0\) is nontwist, we will show \(\iota'(0) = 0\) if and only if \(\tilde{\iota}'(0) = 0\).
    We compute
	\begin{equation*}
		\begin{aligned}
			\iota'(p) &= \frac{1}{n^2} \frac{(n \cdot \omega'(p)) (n^\perp \cdot \omega(p)) - (n \cdot \omega(p)) (n^\perp \cdot \omega'(p))}{(n^\perp \cdot \omega(p))^2} \\
			&= \frac{1}{n^2} \frac{n^T (\omega'(p) \omega(p)^T - \omega(p)\omega'(p)^T) n^\perp}{(n^\perp \cdot \omega(p))^2} \\
			&= \frac{1}{n^2} \frac{n^T A(p) n^\perp}{D(p)}
		\end{aligned}
	\end{equation*}

    The denominator contains the expression \(D(p) = (n^{\perp} \cdot \omega(p))^2\).
	Since \(\omega(0)\) is irrational, \(n^\perp \cdot \omega(0)\) is non-zero, so the denominator will be finite and non-zero in a neighborhood around \(p = 0\).
    Thus, \(\iota(0) = 0\) requires that the numerator is equal to zero.
    
	The numerator contains the \(2 \times 2\) real antisymmetric matrix \(A(p) = \omega'(p) \omega(p)^T - \omega(p)\omega'(p)^T\), which must be proportional to the 90-degree rotation matrix.
    Since
    \begin{equation*}
        n^T \begin{bmatrix}
            0 & -1 \\ 1 & 0
        \end{bmatrix} n^{\perp} = -n^2
    \end{equation*}
	\(\iota'(0) = 0\) can hold if and only if \(A(0) = 0\).
	
	Let \(\tilde{n}\) be the winding numbers associated with \(\tilde{\Gamma}\).
	A similar calculation results in
    \begin{equation*}
        \tilde{\iota}(p) = \frac{1}{\tilde{n}^2} \frac{\tilde{n}^T A(p) \tilde{n}^{\perp}}{\tilde{D}(p)}
    \end{equation*}
    where \(\tilde{D}(p) = (\tilde{n}^{\perp} \cdot \omega(p))^2\).
    \(\tilde{D}(p)\) must also be finite and non-zero in a neighborhood around \(p=0\), so \(\tilde{\iota}'(0) = 0\) if and only if \(A(0) = 0\).
	Thus, \(\tilde{\gamma}_0\) will be nontwist if and only if \(\gamma_0\) is nontwist.
	
	To show \(\tilde{\gamma}_0\) is non-degenerate, we compute
	\begin{equation*}
		\iota''(p) = \frac{1}{n^2} \frac{(n^T A'(p) n^\perp) D(p) - (n^T A(p) n^\perp) D'(p)}{D(p)^2}
	\end{equation*}
    Again, the denominator is finite and non-zero in a neighborhood around \(p=0\).
	Since \(A(0) = 0\), \(\iota''(0)\neq 0\) occurs if and only if \(A'(0) \neq 0\).
	A similar calculation shows \(\tilde{\iota}''(0) \neq 0\) if and only if \(A'(0) \neq 0\), hence \(\tilde{\gamma}_0\) is non-degenerate if and only if \(\gamma_0\) is non-degenerate.
\end{proof}

\bibliographystyle{jpp}

\bibliography{main}

\end{document}